\documentclass[11pt,letterpaper]{article}
\usepackage[utf8]{inputenc}
\usepackage{fullpage,times}
\usepackage{amsmath,amsthm,amsfonts,amssymb,multirow}

\usepackage[procnumbered,ruled,vlined,linesnumbered]{algorithm2e}
\usepackage{thm-restate}
\DontPrintSemicolon
\SetKw{KwAnd}{and}
\SetProcNameSty{textsc}
\SetFuncSty{textsc}

\usepackage{times}
\usepackage{graphicx}
\usepackage{floatpag}
\usepackage{bbold}
\usepackage{enumitem}
\usepackage{subcaption} 
\usepackage{calc}
\usepackage{tikz}
\usetikzlibrary{decorations.markings}
\tikzstyle{vertex}=[circle, draw, inner sep=0pt, minimum size=4pt, fill = black]

\usepackage{graphicx}
\usepackage{tabularx}
\makeatletter
\newcommand{\multiline}[1]{%
  \begin{tabularx}{\dimexpr\linewidth-\ALG@thistlm}[t]{@{}X@{}}
    #1
  \end{tabularx}
}
\makeatother
\usepackage{mathtools}

\makeatletter
\def\BState{\State\hskip-\ALG@thistlm}
\makeatother

\newcommand{\ceil}[1]{\lceil #1 \rceil}
\newcommand{\floor}[1]{\lfloor #1 \rfloor}

\usepackage[compact]{titlesec}
\titlespacing{\section}{0pt}{3ex}{2ex}
\titlespacing{\subsection}{0pt}{2ex}{1ex}
\titlespacing{\subsubsection}{0pt}{0.5ex}{0ex}

\newtheorem{theorem}{Theorem}[section]

\newtheorem{corollary}{Corollary}[section]
\newenvironment{proofof}[1]{{\bf Proof of #1.  }}{\hfill$\Box$}
\newtheorem{definition}{Definition}[section]
\newtheorem{lemma}{Lemma}[section]
\newtheorem{claim}{Claim}
\newtheorem{hypothesis}{Hypothesis}

\newtheorem{observation}{Observation}[section]

\makeatletter
\let\c@fconjecture\c@conjecture
\makeatother

\makeatletter
\let\c@fconj\c@conj
\makeatother

\def \eps {\varepsilon}

\newcommand{\ignore}[1]{}

\def \poly { \text{\rm poly~} }

\title{Induced Cycles and Paths Are Harder Than You Think}

\author{Mina Dalirrooyfard\thanks{Department of Electrical Engineering and Computer Science and CSAIL, MIT. minad@mit.edu}, Virginia Vassilevska Williams\thanks{Department of Electrical Engineering and Computer Science and CSAIL, MIT. virgi@mit.edu}}
\begin{document}
\date{}
\maketitle
\thispagestyle{empty}

\begin{abstract}
The goal of the paper is to give fine-grained hardness results for the Subgraph Isomorphism (SI) problem for fixed size induced patterns $H$, based on the $k$-Clique hypothesis that the current best algorithms for Clique are optimal. 

Our first main result is that for any pattern graph $H$ that is a {\em core}, the SI problem for $H$ is at least as hard as $t$-Clique, where $t$ is the size of the largest clique minor of $H$. This improves (for cores) the previous known results [Dalirrooyfard-Vassilevska W. STOC'20] that the SI for $H$ is at least as hard as $k$-clique where $k$ is the size of the largest clique {\em subgraph} in $H$, or the chromatic number of $H$ (under the Hadwiger conjecture). For detecting \emph{any} graph pattern $H$, we further remove the dependency of the result of [Dalirrooyfard-Vassilevska W. STOC'20] on the Hadwiger conjecture at the cost of a sub-polynomial decrease in the lower bound.

The result for cores allows us to prove that the SI problem for induced $k$-Path and $k$-Cycle is harder than previously known. Previously [Floderus et al. Theor. CS 2015] had shown that $k$-Path and $k$-Cycle are at least as hard to detect as a $\lfloor k/2\rfloor$-Clique. We show that they are in fact at least as hard as $3k/4-O(1)$-Clique, improving the conditional lower bound exponent by a factor of $3/2$. This shows for instance that the known $O(n^5)$ combinatorial algorithm for $7$-cycle detection is conditionally tight.

Finally, we provide a new conditional lower bound for detecting induced $4$-cycles: $n^{2-o(1)}$ time is necessary even in graphs with $n$ nodes and $O(n^{1.5})$ edges. The $4$-cycle is the smallest induced pattern whose running time is not well-understood. It can be solved in matrix multiplication, $O(n^\omega)$ time, but no conditional lower bounds were known until ours. We provide evidence that certain types of reductions from triangle detection to $4$-Cycle would not be possible. We do this by studying a new problem called Paired Pattern Detection.

\end{abstract}
\newpage
\pagenumbering{arabic}
\section{Introduction}


A fundamental problem in graph algorithms, Subgraph Isomorphism (SI) asks, given two graphs $G$ and $H$, does $G$ contain a subgraph isomorphic to $H$? 
While the problem is easily NP-complete, many applications only need to solve the poly-time solvable version in which the pattern $H$ has constant size; this version of SI is often called Graph Pattern Detection and is the topic of this paper. 

There are two versions of SI: {\em induced} and not necessarily induced, {\em non-induced} for short. In the induced version, the copy of $H$ in $G$ must have both edges and non-edges preserved, whereas in the non-induced version only edges need to carry over, and the copy of $H$ in $G$ can be an arbitrary supergraph of $H$.
It is well-known that the induced version of $H$-pattern detection for any $H$ of constant size is at least as hard as the non-induced version (see e.g. \cite{stoc-paper}), and that often the non-induced version of SI has faster algorithms (e.g. the non-induced $k$-independent set problem is solvable in constant time). 

It is well-known that the SI problem for any $k$-node pattern $H$ in $n$-node graphs for constant $k$, can be reduced in linear time to detecting a $k$-clique in an $O(n)$ node graph (see \cite{nevsetvril1985complexity}). Thus the hardest pattern to detect is $k$-clique. A natural question is:

\begin{center}{\em How does the complexity of detecting a particular fixed size pattern $H$ compare to that of $k$-clique?}\end{center}

Let us denote by $C(n,k)$ the best running time for $k$-clique detection in an $n$ node graph. 
When $k$ is divisible by $3$, Ne\v{s}etril and Poljak
\cite{nevsetvril1985complexity} showed that $C(n,k)\leq O(n^{\omega k/3})$ time, where $\omega<2.37286$ \cite{almanv21} is the matrix multiplication exponent. For $k$ not divisible by $3$,
 $C(n,k)\leq O(n^{\omega(\lfloor k/3\rfloor,\lceil k/3\rceil,\lceil (k-1)/3\rceil)})$ time, where $\omega(a,b,c)$ is the exponent of multiplying an $n^a\times n^b$ by an $n^b\times n^c$ matrix.

This $k$-clique running time has remained unchallenged since the 1980s, and a natural hardness hypothesis has emerged (see e.g. \cite{vsurvey}):

\begin{hypothesis}[$k$-clique Hypothesis]
On a word-RAM with $O(\log n)$ bit words,
for every constant $k\geq 3$, $k$-clique  requires $n^{\omega(\lfloor k/3\rfloor,\lceil k/3\rceil,\lceil (k-1)/3\rceil)-o(1)}$ time.
\end{hypothesis}

A ``combinatorial'' version\footnote{``Combinatorial" is not well-defined, but it is a commonly used term to denote potentially practical algorithms that avoid the generally impractical Strassen-like methods for matrix multiplication.} of the hypothesis states that the best combinatorial algorithm for $k$-clique runs in $n^{k-o(1)}$ time. 
Other hypotheses such as the Exponential Time Hypothesis for SAT \cite{ipz1,CalabroIP06} imply weaker versions of the $k$-Clique Hypothesis, namely that $k$-clique requires $n^{\Omega(k)}$ time \cite{ChenCFHJKX04}. We will focus on the {\em fine-grained} $k$-Clique Hypothesis as we are after fine-grained lower bounds that focus on fixed exponents.

Our {\bf goal} is now, for every $k$-vertex pattern $H$, determine a function $f(H)$ such that detecting $H$ in an $n$-vertex graph is at least as hard (in a fine-grained sense, see \cite{vsurvey}) as detecting an $f(H)$-clique in an $n$-vertex graph.
We then say that $H$ is ``at least as hard as $f(H)$-clique''.

Obtaining such results is interesting for several reasons. 
\begin{itemize}
\item First, under the $k$-clique Hypothesis, we would get {\em fine-grained} lower bounds for detecting $H$. This would give us a much tighter handle on the complexity of $H$-detection than, say, results (such as results based on ETH, or \cite{ManurangsiRS21}) that merely provide an $n^{\Omega(k)}$ lower bound which only talks about the {\em growth} of the exponent.

\item Second, knowing the largest size clique that limits the complexity of $H$-pattern detection can allow us to compare between different patterns. The goal is to get to something like: the complexity of $k$-node $H_1$ is like the complexity of $k/10$-clique, whereas the complexity of $k$-node $H_2$ is like the complexity of $k/2$-clique, so $H_2$ seems harder. 

\item Third, this more structural approach uncovers interesting combinatorial and graph theoretic results. For instance, in \cite{stoc-paper} it was uncovered that the colorability of a pattern, and the Hadwiger conjecture can explain the hardness of pattern detection. This is not obvious at all apriori.
\end{itemize}

This approach has been taken by prior work (e.g. \cite{stoc-paper,BlaserKS18,FloderusKLL15}); see the related work section for more background.

%

\section{Our results}
Our contributions are as follows: 
\begin{enumerate}
\item First, we obtain a strengthening of a recent result of \cite{stoc-paper} that implies that the hardness of certain patterns called ``cores'' relates to the size of their maximum {\em clique minor.} This hardness is stronger than what was previously known, as previously only the chromatic number, or the maximum size of a clique {\em subgraph} were known to imply limitations, and both of these parameters are upper-bounded by the clique minor size (under the Hadwiger conjecture, for chromatic number).

\item We then apply the result above to obtain much higher hardness for {\em induced Path and Cycle} detection in graphs: a $k$-path or $k$-cycle contains an independent set of size roughly $k/2$. Thus both $k$-Cycle and $k$-Path were shown \cite{FloderusKLL15} to be at least as hard as $\lfloor k/2\rfloor$-Clique. We raise the hardness to that of $3k/4-O(1)$ clique, thus raising the {\em exponent} of the lower bound running time by a factor of $3/2$. This allows us for instance to obtain a {\em tight} conditional lower bound of $n^{5-o(1)}$ for the running time of combinatorial algorithms for $7$-Clique; an $O(n^5)$ algorithm was obtained by Bl\"{a}ser et al. \cite{blaser2018graph}.

\item Finally, we consider the smallest known case of induced $k$-Cycle whose complexity is not well-understood: induced $4$-Cycle. We provide a new conditional lower bound for the problem in sparser graphs based on a popular fine-grained hypothesis, and also provide some explanation for why reductions from triangle detection to $4$-Cycle have failed so far.

\end{enumerate}

We now elaborate on our results.

\paragraph{New results for core graphs.}
Dalirrooyfard, Vuong and Vassilevska W. \cite{stoc-paper} related the hardness of subgraph pattern detection to the size of the maximum clique or the chromatic number of the pattern.
In particular, they showed that if $H$ has chromatic number $t$, then under the Hadwiger conjecture, $H$ is at least as hard to detect as a $t$-clique. 

The Hadwiger conjecture basically states that the chromatic number of a graph is always at most the largest size of a clique minor of the graph. As the result of \cite{stoc-paper} was already assuming the Hadwiger conjecture, one might wonder if it can be extended to show that every pattern $H$ is at least as hard to detect as an $\eta$-clique, where $\eta$ is the size of the largest {\em clique minor} of $H$.

We first note that such an extension is highly unlikely to work for non-induced patterns: the four-cycle $C_4$ has a $K_3$ (triangle) minor, but a non-induced $C_4$ has an $O(n^2)$ time detection algorithm that does not use matrix multiplication, whereas any subcubic triangle detection algorithm must use (Boolean) matrix multiplication \cite{focsy}. Thus any extension of the result that shows clique-minor-sized clique hardness would either only work for certain types of non-induced graphs, or will need to only work in the induced case.

Here we are able to show that $H$-subgraph pattern detection, even in the non-induced case, is at least as hard as $\eta$-clique, where $\eta$ is the largest clique minor size of $H$, as long as $H$ is a special type of pattern called a {\em core}. Cores include many patterns of interest, including the complements of cycles of odd length.
 We also give several other hardness results, such as removing the dependence on the Hadwiger conjecture from some of the results of \cite{stoc-paper} with only a slight loss in the lower bound. 

We call a subgraph $C$ of a graph $H$ a {\em core} of $H$ if there is a homomorphism $H\rightarrow C$ but there is no homomorphism $H\rightarrow C'$ for any proper subgraph $C'$ of $C$. 
Hell and Ne{\v{s}}et{\v{r}}il~\cite{hell1992core} showed that every graph has a unique core (up to isomorphism), and the core of a graph is an {\em induced} subgraph. We denote the core of a graph $H$ by $core(H)$. A graph which is its own core is called simply a {\em core}.

We prove strong hardness results for cores, relating the hardness of detecting the pattern to the size of its maximum clique minor. We then relate the hardness of detecting arbitrary patterns to the hardness of detecting their cores.

We begin with a theorem that shows hardness for detecting a ``partitioned'' copy of a pattern $H$. Here the vertex set of the host graph $G^*$ is partitioned into $k$ parts, and one is required to detect an induced copy of a $k$-node $H$ so that the image of the $i$th node of $H$ is in the $i$th part of the vertex set of $G^*$. This version of SI is often called Partitioned Subgraph Isomorphism (PSI).  Marx \cite{marx2007can} showed that under ETH, PSI for a pattern $H$ requires at least $n^{\Omega(tw(H)/\log{tw(H)})}$ time where $tw(H)$ is the treewidth of $H$. We give a more fine-grained lower bound for PSI. We provide a reduction from $\eta$-clique detection in an $n$ node graph to PSI for a graph $H$ in an $O(n)$ node host graph, for any $H$ with maximum clique minor of size  $\eta$.

\begin{restatable}{theorem}{psihardness}
\label{thm:psihardness} 
\textbf{(Hardness of PSI)} Let $H$ be a $k$-node pattern with maximum clique minor of size $\eta(H)$, and let $G$ be an $n$-node graph. Then one can construct a $k$-partite $O(n)$-node graph $G^*$ in $O(n^2)$ time such that $G^*$ has a colorful copy of $H$ if and only if $G$ has a clique of size $\eta(H)$.
\end{restatable}

Thus the hardness of Partitioned SI is related to the size of the largest clique minor.
To obtain a bound on the size of the maximum clique minor of any graph we use a result of Thomason \cite{thomason2001extremal} as follows: Let $c(t)$ be the minimum number such that every graph $H$ with $|E(H)|\ge c(t)|V(H)|$ has a $K_t$ minor. Then $c(t)=(\alpha+o(1))t\sqrt{\log{t}}$, where $\alpha\le 0.32$ is an explicit constant. Since for $t=\frac{|E(H)|/|V(H)|}{\sqrt{\log (|E(H)|/|V(H)|)}}$ the above inequality is true, we have the following corollary. 
\begin{corollary}
Let $H$ be a $k$-node $m$-edge pattern. Then the problem of finding a partitioned copy of $H$ in an $n$-node $k$-partite graph is at least as hard as finding a clique of size $\frac{m/k}{\sqrt{\log{m/k}}}$ in an $O(n)$-node graph.
\end{corollary}

Hence, for example if $m=ck^2$ for some constant $c$, then the PSI problem for $H$  cannot be solved in $n^{o\left({\frac{k}{\sqrt{\log{k}}}}\right)}$ time. Thus, for dense enough graphs, we 
improve the lower bound of $n^{\Omega(tw(H)/\log{tw(H)})}$ due to Marx \cite{marx2007can}, since $tw(H)\le k$.

While Theorem~\ref{thm:psihardness} only applies to PSI, one can use it to obtain hardness for SI as well, as long as $H$ is a core. In particular, Marx~\cite{marx2007can} showed that PSI and SI are equivalent on cores. Thus we obtain:

\begin{restatable}{corollary}{corehardnessinsi}
\label{thm:core-hardness}
\textbf{(Hardness of cores in SI)} Let $G$ be an $n$-node $m$-edge graph and let $H$ be a $k$-node pattern with maximum clique minor of size $\eta(H)$. If $H$ is a core, then one can construct a graph $G^*$ with at most $O(n)$ vertices in $O(m+n)$ time such that $G^*$ has a subgraph isomorphic to $H$ if and only if $G$ has a $\eta(H)$-clique as a subgraph.
\end{restatable}

As the complements of odd cycles are cores with a clique minor of size at least $\floor{3k/4}$, for $C_k$ when $k$ is odd, we immediately obtain a lower bound of $C(n,\floor{3k/4})$ for $C_k$ detection. When $k$ is even, more work is needed.



Corollary~\ref{thm:core-hardness} applies to the non-induced version of SI. We obtain a stronger result for the induced version in terms of the $k$ and the size of the largest clique subgraph.

\begin{restatable}{corollary}{corehardnessgeneral}
\label{cor:coreLB}
\textbf{(Hardness for induced-SI for cores)}
Let $H$ be a $k$-node pattern which is a core. Suppose that $w(H)$ is the size of the maximum clique in $H$. Then detecting $H$ in an $n$-node graph as an induced subgraph is at least as hard as detecting a clique of size $\max\{\ceil{\sqrt{(k+2w(H))/2}},\ceil{\sqrt{k/1.95}\}}$.
\end{restatable}

For comparison, the result of \cite{stoc-paper} shows that
 non-induced SI for any $k$-node $H$ is at least as hard as detecting a clique of size $\sqrt{k}$, but the result is conditioned on the Hadwiger conjecture. 
 Corollary \ref{cor:coreLB} is the strongest known clique-based lower bound result for $k$-node {\em core} $H$ that is not conditioned on the Hadwiger conjecture.




Our next theorem relates the hardness of detecting a pattern to the hardness of detecting its core.

\begin{restatable}{theorem}{generalhardness}
\label{thm:general-hardness}
Let $G$ be an $n$-node $m$-edge graph and let $H$ be a $k$-node pattern. Let $C$ be the core of $H$. Then one can construct a graph $G^*$ with at most $O(n)$ vertices in $O(n^2)$ time such that $G^*$ has a subgraph isomorphic to $H$ if and only if $G$ has a subgraph isomorphic to $C$, with high probability\footnote{with probability $1/\poly{n}$}.
\end{restatable}

One consequence of Theorem~\ref{thm:general-hardness} and Corollary \ref{cor:coreLB} is that induced-SI for \emph{any} pattern $H$ of size $k$ is at least as hard as detecting a clique of size $\ceil{k^{1/4}/1.39}$. Note that this is the first lower bound for induced SI that is {\bf only} under the $k$-clique hypothesis.

\begin{restatable}{corollary}{inducedsihardnessone}
\label{cor:inducedhardnessone}
\textbf{(Hardness of Induced-SI)}
For any $k$-node pattern $H$, detecting an induced copy of $H$ in an $n$-node graph is at least as hard as detecting a clique of size $\ceil{k^{1/4}/1.39}$ in an $O(n)$ graph.
\end{restatable}






\paragraph{Hardness for induced cycles and paths.}
We now focus on $k$-paths $P_k$ and $k$-cycles $C_k$ for fixed $k$ and provide highly improved fine-grained lower bounds for their detection under the $K$-clique Hypothesis (for $k$ larger than some constant). The results can be viewed as relating how close induced paths and cycles are to cliques. Our techniques for proving our results can be of independent interest and can potentially be implemented to get stronger hardness results for other classes of graphs. 

\begin{table}[h]
\caption{Known upper and lower bounds for paths and cycles. The algorithms for $P_3$ and $C_3$ are folklore,  for $P_4$ from \cite{4path}, for $C_4$  from \cite{williams2014finding}, and for $C_k$ and $P_k$ for $k\geq 5$ from \cite{BlaserKS18}. The old lower bounds are all from \cite{FloderusKLL15}; our new lower bounds appear in Theorem~\ref{thm:hardnesscyclepath}.}
\centering
\begin{tabular}{|c|c|c|c|c|}
\hline
Pattern & Runtime  & Lower Bound & Comb. Runtime & Comb. Lower Bound\\
\hline\hline
$P_3$, $P_4$ & $O(m+n)$   & & $O(m+n)$  & \\
\hline
$C_4$ & $\tilde{O}(n^\omega)$  &   & $O(n^3)$  & \\
\hline
$C_3$, $C_5,$ $P_5$ & $\tilde{O}(n^\omega)$   & $n^{\omega-o(1)}$  & $O(n^3)$  & $n^{3-o(1)}$ \\
\hline
$C_6,P_6$ & $C(n,4)$  & $n^{\omega-o(1)}$ & $O(n^4)$ & $n^{3-o(1)}$ \\
& $=O(n^{\omega(2,1,1)})$ &   &   & \\
\hline
$P_7$ & $C(n,6)$ & $n^{\omega-o(1)}$ & {\bf $O(n^{5})$} & $n^{3-o(1)}$ \\
  & $=O(n^{2\omega})$ &  &   &  \\
\hline
$C_7$ & $C(n,6)$ & \textcolor{red}{$C(n,5)$} {\bf [new]} & {\bf $O(n^{5})$} & \textcolor{red}{$n^{5-o(1)}$} {\bf [new]} \\
  & $=O(n^{2\omega})$ & $n^{\omega-o(1)}$  [old] &   & $n^{3-o(1)}$ [old] \\
\hline 
$C_k$, $P_k$, $k\geq 8$ & $C(n,k)$ & \textcolor{red}{$C(n,\lfloor 3k/4\rfloor-2)$} {\bf [new]} & $O(n^{k-2})$  & \textcolor{red}{$n^{\lfloor3k/4\rfloor-2-o(1)}$} {\bf [new]}\\
$C_k$, odd $k$ & $C(n,k)$ & \textcolor{red}{$C(n,\lfloor 3k/4\rfloor)$} {\bf [new]} & $O(n^{k-2})$  & \textcolor{red}{$n^{\lfloor3k/4\rfloor-o(1)}$} {\bf [new]}\\
 &   & $C(n,\lceil k/2 \rceil)$  [old] & & $n^{\lceil k/2 \rceil-o(1)}$ [old]\\
\hline
\end{tabular}
\label{tab:paths}
\end{table}


The fastest known algorithms for finding induced cycles or paths on $k$ nodes can be found in Table \ref{tab:paths}. For larger $k$, the best known algorithms are either the $k$-clique running time $C(n,k)$, or an $O(n^{k-2})$ time combinatorial algorithm by \cite{BlaserKS18}. For $k\leq 7$, slightly faster algorithms are known.

The best known conditional lower bounds so far \cite{FloderusKLL15} under the $k$-clique hypothesis stem from the fact that the complement of $C_k$ contains a $\lfloor k/2 \rfloor$-clique, and the complement of $P_k$ contains a $\lceil k/2\rceil$-clique.
These lower bounds show that the best known running time of $O(n^\omega)$ for $C_5$ and $P_5$ are likely optimal. Unfortunately, for larger $k$, these lower bounds are far from the best known running times.

We obtain polynomially higher lower bounds, raising the lower bound exponent from roughly $k/2$ to roughly $3k/4$.

\begin{restatable}{theorem}{hardnesscyclepath}
\label{thm:hardnesscyclepath}
\textbf{(Hardness of $P_k$ and $C_k$)}
Let $H$ be the complement of a $P_k$ or the complement of a $C_k$. Suppose that $t$ is the size of the maximum clique minor of $H$. Then the problem of detecting $H$ in an $O(n)$-node graph is at least as hard as finding a $(t-2)$-clique in an $n$-node graph. If $k$ is odd, then detecting an induced $C_k$ is at least as hard as finding a $t$-clique.
\end{restatable}

The largest \emph{clique minor}\footnote{A $t$-clique minor of a graph $H$ is a decomposition of $H$ into $t$ connected subgraphs such that there is at least one edge between any two subgraphs.} of the complement of $C_k$ has size $\lfloor 3k/4\rfloor$ and of
the complement of 
$P_k$ has size at least $\lfloor (3k+1)/4\rfloor$.

Table \ref{tab:paths} summarizes our new lower bounds. Aside from obtaining a much higher conditional lower bound, our result shows that the best known combinatorial algorithm for $C_7$ detection is tight, unless there is a faster combinatorial algorithm for $5$-clique detection.
For algorithms that may be non-combinatorial, our lower bound for $C_7$ is at least $\Omega(n^{4.08})$ assuming that the current bound for $5$-clique is optimal.

\paragraph{The curious case of Four-Cycle.} 
The complexity of SI for all patterns on at most $3$ nodes in $n$-node graphs is well-understood, both in the induced and non-induced case: all patterns except the triangle and (in the induced case) the independent set can be detected in $O(n^2)$ time, whereas the triangle (and independent set in the induced case) can be detected in $O(n^\omega)$ time where $\omega<2.373$ is the exponent of matrix multiplication \cite{almanv21}. The dependence on (Boolean) matrix multiplication for triangle detection was proven to be necessary \cite{focsy}.

Table \ref{fig:smallgraphs} gives the best known algorithms and conditional lower bounds for induced SI for all $4$-node patterns. In the non-induced case, the change is that, except for the $4$-clique $K_4$, the diamond, co-claw and the paw whose runtimes and conditional lower bounds stay the same, all other patterns can be solved in $O(n^2)$ time.

\begin{figure}[ht]\centering
\includegraphics[width=8cm]{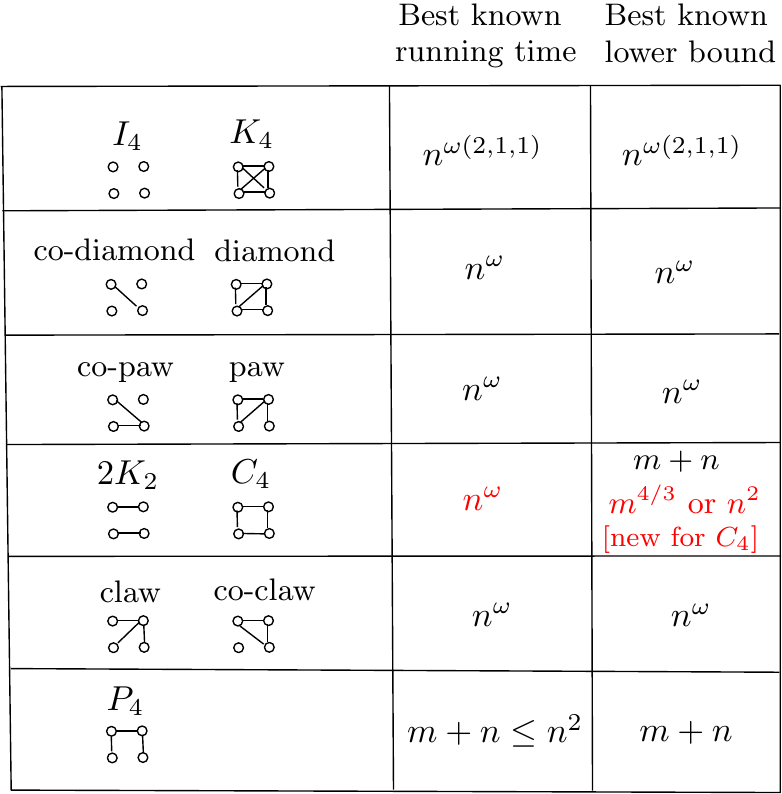}
\caption{A list of all induced $4$-node patterns with their best known algorithms and conditional lower bounds. All runtime bounds are within polylogarithmic factors, and can be found in \cite{williams2014finding} and \cite{EG04}. All lower bounds except for the new one in red are under the $k$-clique hypothesis for $k=3$ or $4$, and every lower bound $l$ as a function of $m$ or $n$ should be interpreted as $l^{1-o(1)}$. The new lower bound is under the $3$-uniform $4$-hyperclique hypothesis, only holds for $C_4$ and is in Theorem~\ref{thm:c4hard}. All upper and lower bounds are tight except that for $C_4$ and its complement.}
\label{fig:smallgraphs}
\end{figure}

All conditional lower bounds in Table \ref{fig:smallgraphs} are tight, except for the curious case of the induced $4$-Cycle $C_4$. Non-induced $C_4$ can famously be detected in $O(n^2)$ time (see e.g. \cite{non-induced-4cycle}). Meanwhile, the fastest algorithm for induced $C_4$ runs in $O(n^\omega)$ time (see e.g. \cite{williams2014finding}). 
There is no non-trivial lower bound known for $C_4$ detection (except that one needs to read the graph), and obtaining a higher lower bound or a faster algorithm for $C_4$ has been stated as an open problem several times (see e.g. \cite{EschenHSS11}).

The induced $4$-cycle is the smallest pattern $H$ whose complexity is not tightly known, under any plausible hardness hypothesis.

We make partial progress under the popular $3$-Uniform $4$-Hyperclique Hypothesis (see e.g. \cite{LincolnWW18,AbboudBDN18}) that postulates that hyperclique on $4$ nodes in an $n$ vertex $3$-uniform hypergraph cannot be detected in $O(n^{4-\eps})$ time for any $\eps>0$, in the word-RAM model of computation with $O(\log n)$ bit words. The believability of this hyperclique hypothesis is discussed at length in \cite{LincolnWW18} (see also \cite{AbboudBDN18}); one reason to believe it is that refuting it would imply improved algorithms for many widely-studied problems such as Max-$3$-SAT \cite{Williams05}.

\begin{theorem}
Under the $3$-Uniform $4$-Hyperclique Hypothesis, there is no $O(m^{4/3-\eps})$ time or $O(n^{2-\eps})$ time algorithm for $\eps>0$ that can detect an induced $4$-cycle in an $n$-node, $m$-edge undirected graph.\label{thm:c4hard}
\end{theorem}

While our result conditionally rules out, for instance, a linear time (in the number of edges) algorithm for induced $C_4$, it does not rule out an $O(n^2)$ time algorithm for induced $C_4$ in dense graphs since the number of edges in the reduction instance is $\Theta(n^{3/2})$ in terms of the number of nodes $n$.
Ideally, we would like to have a reduction from triangle detection to induced $C_4$-detection, giving evidence that $n^{\omega-o(1)}$ time is needed. Our Theorem does show this if $\omega=2$, but we would like the reduction to hold for any value of $\omega$, and for it to be meaningful in dense graphs. Note that even if $\omega=2$, a reduction from triangle detection would be meaningful, as it would say that a practical, combinatorial algorithm would be extremely difficult to obtain (or may not even exist).



All known reductions from $k$-clique to SI for other patterns $H$ (e.g. \cite{FloderusKLL15,stoc-paper,Lin15}) work equally well for non-induced SI. In particular, in the special case when $H$ is bipartite, such as when $H=C_4$, the host graph also ends up being bipartite (e.g. \cite{Lin15} for bicliques, and  \cite{FloderusKLL15,stoc-paper} more generally).


Unfortunately such reductions are doomed to fail for $C_4$. 
In bipartite graphs and more generally in triangle-free graphs, any non-induced $C_4$ is an induced $C_4$. Of course, any hypothetical fine-grained reduction from triangle detection to non-induced $C_4$ detection in triangle-free graphs, combined with the known $O(n^2)$ time algorithm for non-induced $C_4$ would solve triangle detection too fast.

The difference between induced $C_4$ and non-induced $C_4$ is that the latter calls for detecting {\em one} of the three patterns: $C_4,$ diamond or $K_4$. Could we have a reduction from triangle detection to induced $C_4$-detection in a graph that is {\em not} triangle-free, but is maybe $K_4$-free?
In order for such a reduction to work, it must be that detecting one of $\{C_4,K_4\}$ is computationally hard.

We show that such reductions are also doomed. We provide a fast combinatorial algorithm that detects one of $\{C_4,H\}$ for any $H$ that contains a triangle. The algorithm in fact runs faster than the current matrix multiplication time, which (under the $k$-clique Hypothesis) is required for detecting any $H$ containing a triangle. Thus, any tight reduction from triangle detection to induced $C_4$ must create instances that contain {\em every} induced $4$-node $H$ that has a triangle.

\begin{restatable}{theorem}{cfourgeneral}
\label{thm:c4-general}
For any $4$-node graph $H$ that contains a triangle, 
detecting one of $\{C_4,H\}$ as an induced subgraph of a given $n$-node host graph can be done in $O(n^{7/3})$ time.
If $H$ is not a diamond or $K_4$, then $H$ or $C_4$ can be detected 
 $O(n^2)$ time.
\end{restatable}

The only case of Theorem \ref{thm:c4-general} that was known is that for $\{C_4,$ diamond $\}$. Eschen et al. \cite{EschenHSS11} considered the recognition of $\{C_4,$ diamond $\}$-free graphs and gave a combinatorial $O(n^{7/3})$ time algorithm for the problem. We show a similar result for every $H$ that contains a triangle. 

The $C_4$ OR $H$ problem solved by our theorem above is a special case of the subgraph isomorphism problem in which we are allowed to return one of a {\em set of possible patterns}.  This version of SI is a natural generalization of non-induced subgraph isomorphism in which the set of patterns are all supergraphs of a pattern. This generalized version of SI
has practical applications as well. Often computational problems needed to be solved in practice are not that well-defined, so that for instance
you might be looking for something like a matching or a clique, but maybe you are okay with extra edges or some edges missing.
In graph theory applications related to graph coloring, one is often concerned with $\{H,F\}$-free graphs for various patterns $H$ and $F$ (e.g. \cite{KarthickM18,DabrowskiP18,DabrowskiDP17}). Recognizing such graphs is thus of interest there as well. 
 We call the problem of detecting one of {\em two} given induced patterns, ``Paired Pattern Detection''. 

Intuitively, if a set of patterns all contain a $k$-clique, then returning at least one of them should be at least as hard as $k$-clique. While this is intuitively true, proving it is not obvious at all. In fact, until recently \cite{stoc-paper}, it wasn't even known that if a single pattern $H$ contains a $k$-clique, then detecting an induced $H$ is at least as hard as $k$-clique detection. We are able to reduce $k$-clique in a fine-grained way to ``Subset Pattern Detection" for any subset of patterns that all contain the $k$-clique as a subgraph\footnote{Our reduction works in the weaker non-induced version and so it works for the induced version as well.}.

\begin{restatable}{theorem}{setdetectioncore}
\label{thm:set-detection-cor}
Let $S$ be a set of patterns such that every $H\in S$ contains a $k$-clique. Then detecting whether a given graph contains some pattern in $S$ is at least as hard as $k$-clique detection.
\label{thm:cliqueset}
\end{restatable}

While having a clique in common makes a subset of patterns hard to detect, intuitively, if several patterns are very different from each other, then detecting one of them should be easier than detecting each individually. We make this formal for Paired Pattern Detection in $n$ node graphs for $k\leq 4$ as follows:

\begin{itemize}
\item Paired Pattern Detection is in $O(n^2)$ time for {\em every pair} of $3$ node patterns. Moreover, for all but two pairs of patterns, it is actually in linear time.
\item Paired Pattern Detection for any pair of $4$-node patterns is in $\tilde{O}(n^\omega)$ time, whereas the fastest known algorithm for $4$-clique runs in supercubic, $O(n^{\omega(1,2,1)})$ time where $\omega(1,2,1)\leq 3.252$ \cite{GallU18} is the exponent of multiplying an $n\times n^2$ by an $n^2\times n$ matrix.
\item There is an $O(n^2)$ time algorithm that solves Paired Pattern Detection for $\{H,\bar{H}\}$ for any $4$-node $H$, where $\bar{H}$ is the complement of $H$. 
\end{itemize}

 The last bullet is a generalization of an old Ramsey theoretic result of Erd\"{o}s and Szekeres~\cite{erdosramsey}
made algorithmic by Boppana and Halld{\'{o}}rsson~\cite{BoppanaH92}. The latter shows that in linear time for any $n$-node graph, one can find either a $\log(n)$ size independent set or a $\log(n)$ size clique.
 Thus, for every constant $k$ and large enough $n$, there is a linear time algorithm that either returns a $k$-clique or an $I_k$.
 
We note that our generalization for $\{H,\bar{H}\}$ cannot be true in general for $k\geq 5$: both $H$ and its complement \footnote{For $k=5$, consider $H$ to be a triangle and two independent nodes. Both $H$ and its complement contain a triangle.} can contain a clique of size $\lceil k/2\rceil\geq 3$, and thus by our Theorem~\ref{thm:cliqueset}, their Paired Pattern Detection is at least as hard as $\lceil k/2\rceil$-clique, and thus is highly unlikely to have an $O(n^2)$-time algorithm.

\subsection{Related work}
There is much related work on the complexity of graph pattern detection in terms of the {\em treewidth} of the pattern. Due to the Color-Coding method of Alon, Yuster and Zwick \cite{alon1995color}, it is known that if a pattern $H$ has treewidth $t$, then detecting $H$ as a {\em non-induced} pattern can be done in $O(n^{t+1})$ time.
This implies for instance that non-induced $k$-paths and $k$-cycles can be found in $2^{O(k)}\textrm{poly}(n)$ time.

Marx \cite{marx2007can} showed that there is an infinite family of graphs of unbounded treewidth so that under ETH, (non-induced) SI on these graphs requires $n^{\Omega(t/\log t)}$ time where $t$ is the treewidth of the graph. Recently, Bringmann and Slusallek \cite{BringmannS21} showed that under the Strong ETH, for every $\eps>0$, there is a $t$ and a pattern $H$ of treewidth $t$ so that detecting $H$ cannot be done in $O(n^{t+1-\eps})$ time. That is, for some non-induced patterns, $n^{t+1}$ is essentially optimal.

In the induced case, many patterns are also easier than $k$-clique, e.g. for $k=3,4,5,6$, any $H$ that is not the $k$-independent set or the $k$-clique can be found in the current best running time $C(n,k-1)$ for $(k-1)$-clique \cite{williams2014finding,BlaserKS18,stoc-paper}.
For $k\geq 7$, Bl\"{a}ser et al.~\cite{BlaserKS18} showed the weaker result that all $k$-node $H$ that are not the clique or independent set can be detected in $O(n^{k-1})$ time combinatorially, whereas the best known combinatorial algorithms for $k$-clique run in $n^{k-o(1)}$ time.

For induced pattern detection for patterns of size $k\ge 8$, the best algorithm for almost all of the patterns has the same running time as $k$-clique detection. If we only resort to combinatorial algorithms there is a slight improvement: any pattern that is not a clique or independent set can be detected in $O(n^{k-1})$ time \cite{BlaserKS18}.

Manurangsi, Rubinstein and Schramm~\cite{ManurangsiRS21} formulated a brand new hypothesis on the hardness of planted clique. This new hypothesis implies many results that are not known to hold under standard hypotheses such as ETH or Strong ETH, including that for every $k$-node $H$, its induced pattern detection  problem requires $n^{\Omega(k)}$ time. While identifying new plausible hypotheses is sometimes worthwhile, our work strives to get results under standard widely-believed hypotheses, and to uncover combinatorial relationships between $H$-pattern detection and clique-detection, as cliques are the hardest patterns to detect.

Note that the results of Marx~\cite{marx2007can} and Bringmann and Slusallek~\cite{bringmann2021current} show hardness for specific classes of patterns, whereas the the results of Dalirrooyfard et al.~\cite{stoc-paper}, Manurangsi et al.~\cite{ManurangsiRS21} and this paper aim to determine hardness for \emph{any} $k$-node pattern. Our paper primarily focuses on giving lower bounds for fixed patterns such as $C_4,C_7$ etc., whereas the focus of \cite{ManurangsiRS21} is more asymptotic.

\subsection{Organization of the paper} In Section \ref{sec:technical-overview} we give a high level overview of our techniques, and a comparison to the past techniques. 
In Section \ref{sec:prelim} we give the necessary definitions. In Section \ref{sec:lower-bounds} we first state our hardness result for PSI (Theorem \ref{thm:psihardness}) in subsection \ref{sec:hardnessofpsi}, and then in subsection \ref{sec:corehardness} we state our hardness result for SI (Theorem \ref{thm:general-hardness}). Finally, in subsection \ref{sec:paths-and-cycles} we show hardness for paths and cycles (Theorem \ref{thm:hardnesscyclepath}).
We state our results on Paired Pattern Detection in Section \ref{sec:paired} by first showing hardness for Subset Pattern Detection (Theorem \ref{thm:set-hardness}) and then we state our algorithmic results. In section \ref{sec:c4lb} we state our lower bound for induced four cycle detection from $3$-Uniform $4$-Hyperclique Hypothesis.

\label{sec:prelim}
For an integer $k$, let $P_k,C_k,K_k$ and $I_k$ be the path, cycle, clique and independent set on $k$ nodes.

Let $G$ be a graph and $S\subseteq V(G)$ be a subgraph of it. For every node $v\in V(G)$, define $N_S(v)$ to be the neighbors of $v$ in $S$. Define $d_S(v)=|N_S(v)|$.

A $k$-partite graph $G$ can be decomposed into $k$ partitions $G_1,\ldots,G_k$ where each $G_i$ is an independent set. For a pattern $H$ of size $k$ with vertices $v_1,\ldots,v_k$, we say that a graph $G$ is $H$-partite if it is a $k$-partite graph with $G_1,\ldots,G_k$ as its partitions such that there is no edge between $G_i$ and $G_j$ if $v_iv_j$ is not an edge in $H$. 

Let $G$ be an $H$-partite subgraph for a pattern $H$. We say that subgraph $H'$ of $G$ is a \textit{colorful} copy of $H$ if $H'$ has exactly one node in each partition of $G$. Note that if the vertices of $H'$ are $u_1,\ldots,u_k$ where $u_i$ is a copy of $v_i$ for all $i$, then $u_i$ must be in $G_i$ for all $i$\footnote{Note that this statement and many more in the paper are true up to automorphisms}. This is because for every $i,j$ where $v_iv_j$ is an edge, there must be an edge between the vertex of $H'$ that is in $G_i$ and the vertex of $H'$ that is in $G_j$. Otherwise, the number of edges of $H'$ is going to be smaller than the number of edges of $H$. 

For a set of patterns $S$, by (induced) \textit{$S$-detection} we mean finding a (induced) copy of one of the patterns in $S$, or indicating that there is no copy of any of the patterns in $S$.

Let $f:\{1,\ldots,c\}\rightarrow V(H)$ be a \textit{proper} coloring of the graph $H$ if the color of any two adjacent nodes is different. Let the \textit{chromatic number} of a graph $H$ be the smallest number $c$ such that there exists a proper coloring of $H$ with $c$ colors. We say that a graph $H$ is \textit{color critical} if the chromatic number of $H$ decreases if we remove any of its nodes. 

We call the subgraph $C$ of a graph $H$ a {\em core} of $H$ if there is a homomorphism $H\rightarrow C$ but there is no homomorphism $H\rightarrow C'$ for any proper subgraph $C'$ of $C$. Recall that a graph which is its own core is called simply a {\em core}. Moreover, any graph has a unique core up to isomorphisms, and the core of a graph is an {\em induced} subgraph of it \cite{hell1992core}.

\section{Technical Overview}\label{sec:technical-overview}
Here we give high level overview of our techniques. 
To understand our lower bounds for $k$-node patterns, we should first give an overview of the techniques used in \cite{stoc-paper}. In their first result \cite{stoc-paper} shows that if $H$ is $t$-chromatic and has a $t$-clique, then it is at least as hard to detect as a $t$-clique.
\paragraph{Reduction (1) \cite{stoc-paper}.} To prove the result of \cite{stoc-paper}, suppose that we want to reduce detecting a $t$-clique in a host graph $G=(V,E)$ to detecting $H$ in a graph $G^*$ built from $G$ and $H$. We build $G^*$ by making a copy $G^*_h$ of the vertices of $G$ for each node $h\in H$ as an independent set. Then if $hh'\in E(H)$, we put edges between $G^*_h$ and $G^*_{h'}$ using $E$: if $uw\in E$, then we connect the copy of $u$ in $G^*_h$ to the copy of $w$ in $G^*_{h'}$. Note that we have edges between $G^*_h$ and $G^*_{h'}$ if and only if $hh'$ is an edge and this enforces an encoding of $H$ in $G^*$ (we refer to $G^*$ as being $H$-partite).

To show that this reduction works, first suppose that there is a $t$-clique $\{v_1,\ldots,v_t\}$ in $G$. To prove that there is a $H$ in $G^*$, we consider a $t$ coloring of the vertices of $H$, and then we pick a copy of $v_i$ from $G^*_h$ if $h$ has color $i$. Using  the structure of $G^*$ and the fact that no two adjacent nodes in $H$ have the same color, one can show these $|H|$ nodes form a copy of $H$. For the other direction, suppose that there is a copy of $H$ inside $G^*$. This copy contains a $t$-clique $\{w_1,\ldots,w_t\}$. Since each $G^*_h$ is an independent set, no two nodes of the $t$-clique are in the same $G^*_h$. Moreover, the edges in $G^*$ mimic the edges in $G$ and this is sufficient to conclude that no two nodes of the $t$-clique are copies of the same node in $G$, and the original nodes in $G$ that the nodes $w_i$ are the copies of, form a $t$-clique. See Figure \ref{fig:reduction1}.

\begin{figure}
    \centering
    \includegraphics[width=\linewidth]{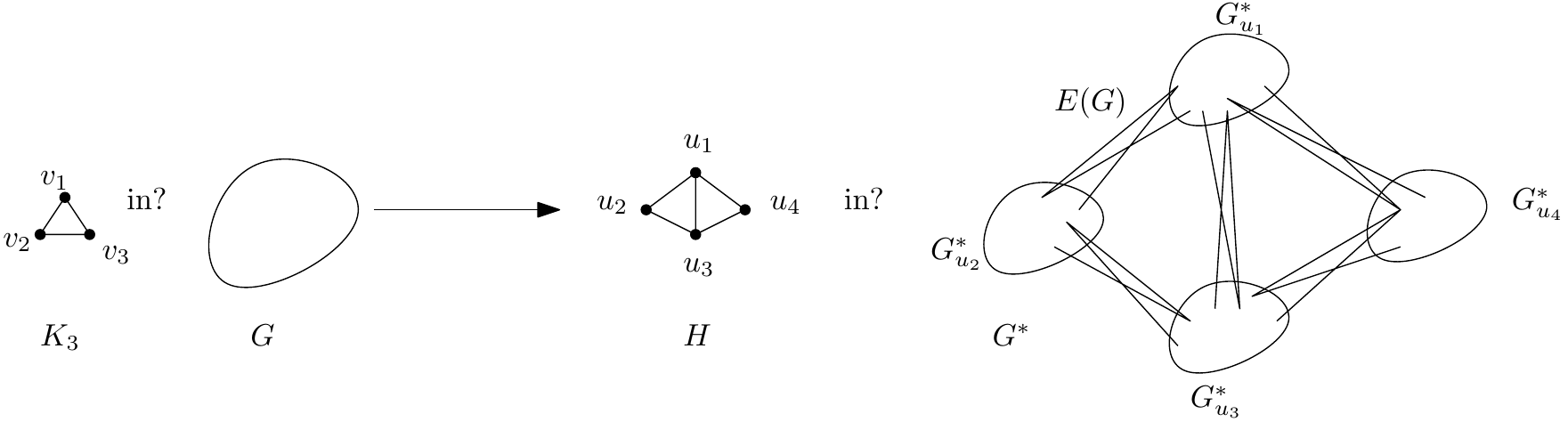}
    \caption{Reduction (1) construction: Reducing $K_3$-detection in $G$ to $H$-detection in $G^*$ for $H$ being the Diamond. The pattern $H$ is $3$ colorable and has a $3$-clique.}
    \label{fig:reduction1}
\end{figure}

Now we show how we modify this reduction to prove our first result, Theorem \ref{thm:psihardness}.
\paragraph{Reduction (2).}
We prove that if the size of the largest clique minor of a pattern $H$ is $\mu(H)$, then detecting a $\mu(H)$-clique in a graph $G$ can be reduced to detecting a colorful copy of $H$ in a graph $G^*$ that is constructed from $G$ and $H$ (Theorem \ref{thm:psihardness}). Reduction (1) above is good at catching cliques that are in the pattern $H$, but $H$ might not have a clique of size $\mu(H)$ in it, so we need a way to encode the clique minor of $H$ in $G^*$ so that it translates to a clique in $G$. To do that, we use a second method to put edges between $G^*_h$ and $G^*_{h'}$ when $hh'$ is an edge. We consider a clique minor of $H$ of size $\mu(H)$. Note that the clique minor partitions the vertices of $H$ into connected subgraphs with at least one edge between every two partition. Now if $hh'$ is an edge in $H$ and $h$ and $h'$ are in the same partition in the clique minor, we want to treat them as one node. So we put a ``matching" between $G^*_h$ and $G^*_{h'}$: for any node $v\in V(G)$, we put an edge between the copy of $v$ in $G^*_h$ and the copy of $v$ in $G^*_{h'}$. This way we show that whenever there is a colorful copy of $H$ in $G^*$, if $h$ and $h'$ are in the same partition of the clique minor of $H$, the vertices that are selected from $G^*_h$ and $G^*_{h'}$ must be copies of the same node in $G$. This means that each clique minor partition of $H$ represents one node in $G$. For $h$ and $h'$ that are not in the same clique minor partition, we put edges between $G^*_h$ and $G^*_{h'}$ the same as Reduction (1) (mimicking $E$). Using the rest of the properties of the construction, we show that the set of nodes that each clique minor partition represents are all distinct, and they form a $\mu(H)$-clique in $G$. 

Note that \cite{stoc-paper} uses the idea in Reduction 2 (a second method to define the edges of $G^*$) in a separate result. However the use of clique minors in \cite{stoc-paper} is indirect; it is coupled with the chromatic number and proper colorings of $H$, and in our results we directly use clique minors without using any other properties, thus avoiding the Hadwiger conjecture. 


Another thing to note about Reduction (2) is that we are reducing a clique detection problem to a ``partitioned" subgraph isomorphism (PSI) problem. The reduction immediately fails if one removes the partitioned constraint. The reason is that we no longer can assume that if the reduction graph $G^*$ has a copy of $H$, then the nodes are in different vertex subsets $G^*_h$. If $G^*$ has a copy of $H$ and two nodes $v,u$ of this copy are in one vertex subset $G^*_h$, then we don't know if $v$ and $u$ are adjacent in $G$ or not. This can get in the way of finding a clique of the needed size in $G$. So if we want to get any result stronger than Reduction (1) for SI (and not PSI), we need to add new ideas. We introduce some of these new ideas below.

\paragraph{Reduction (3): paths and cycles.}
In Theorem \ref{thm:hardnesscyclepath} we show that if $H$ is the complement of a cycle or a path, then we can reduce detecting a $(\mu(H)-2)$-clique in a graph $G$ to detecting a copy of $H$ in a graph $G^*$ constructed from $G$ and $H$. 

As mentioned above, removing the partitioned constraint from reduction (2) doesn't directly work. However, when the graph is a core, it does work, and that is because PSI and SI are equivalent for cores \cite{marx2007can}. 
When $H$ is a core, there is only one homomorphism from $H$ to itself, which means that there is only one type of ``embedding" of $H$ in the reduction graph $G^*$, and it is the embedding with exactly one vertex in each vertex subset $G^*_h$ of $G^*$. However, when $H$ is not a core, there can be multiple embeddings of $H$ in $G^*$, and these embeddings do not necessarily result in finding a copy of a $t$-clique in $H$, for $t\approx \mu(H)$.

In order to solve this issue of multiple embeddings, we ``shrink" some of the vertex subsets ($G^*_h$s) of the reduction graph $G^*$. More formally, we replace some of these subsets in $G^*$ by a single vertex. We do it in such a way that the only embedding of $H$ in $G^*$ is the one with exactly one vertex in each subset. This way, the rest of the argument of Reduction (2) goes through. There is a cost to shrinking these subsets: shrinking more subsets results in reducing the size of the clique that we reduce from. So the harder part of this idea is to carefully decide which partitions to shrink, so that we only lose a small constant in the size of the clique detection problem that we are reducing from. 

Recall that in Reduction (2) we consider a $\mu(H)$ clique minor of $H$ which partitions the vertex set of $H$ into $\mu(H)$ connected subgraphs. Here we observe that for $H$ that is the complement of a path or a cycle, we can select two particular partitions of the clique minor, and shrink vertex subsets $G^*_h$ for vertices $h$ that belong to one of these two partitions. This way we eliminate all the unwanted embeddings of the pattern $H$ in $G^*$, and reduce $(\mu(H)-2)$-clique detection in $G$ to $H$ detection in $G^*$. We note that the techniques in Reduction (3)  are of independent interest and can be potentially used for other graph classes.


We now move on to our next reduction.
\paragraph{Reduction (4).}
Our next main result is Theorem \ref{thm:general-hardness}, which states that if $C$ is the core of the pattern $H$, then detecting $C$ in a graph $G=(V,H)$ can be reduced to detecting $H$ in a graph $G^*$ which is constructed from $G$ and $H$.

First note that Reduction (1) doesn't directly work here. This is because if $G$ has a copy of $C$, we have no immediate way of finding a copy of $H$ in $G^*$. Recall that in Reduction (1) we used a coloring property of $H$ to do this.

As a first attempt to such a reduction, one might use the following idea of Floderus et al. \cite{FloderusKLL15}. They showed that any pattern that has a $t$-clique that is disjoint from all the other $t$-cliques in the pattern is at least as hard as $t$-clique to detect. 
Here we explain their idea in the context of reducing $core(H)$-detection to $H$ detection. Let $C'$ be a copy of $C$ in $H$. The idea is to build the reduction graph of Reduction (1) using $C'$ as the pattern, and to add the rest of the pattern $H$ to it.
More formally, 
for any node $h$ in $C'$, let $G^*_h$ be a copy of $V$, the set of vertices of $G$. Put edges between $G^*_h$ and $G^*_{h'}$ same as before if $hh'$ is an edge. Call this graph $G^*_{C'}$. To complete the construction of $G^*$, add a copy of the subgraph $H\setminus C'$ to $G^*_{C'}$, and connect a vertex $h'$ in this copy to all the nodes in $G^*_h$ for $h\in C'$ if $hh'$ is an edge in $H$. 

The reason we construct the Reduction (1) graph on $C'$ is that if $G$ has a copy of $C$, then we can find a copy of $C$ in $G^*_{C'}$ using the arguments in Reduction (1). This copy of $C$ and all the vertices in $G^*\setminus G^*_{C'}=H\setminus C'$ form a copy of $H$. For the other direction, suppose that there is a copy $H^*$ of $H$ in $G^*$. We hope that this copy contains a subgraph $C''$ that is completely inside $G^*_{C'}$, so that then this leads us to a copy of $C$ in $G$ using the properties of $G^*_{C'}$ and the fact that $C$ is a core. However, such a construction cannot guarantee this, and in fact there might be no copies of $C$ in $H^*$ that are completely in $G^*_{C'}$.  

So we need to find a subgraph $H'$ of $H$, so that if we build the reduction graph $G^*_{H'}$ of Reduction (1) on it, it has the property that if $G $ has a copy of $C$, then we can find a copy of $H'$ in $G^*_{H'}$. 

To do this, we simplify and use an idea of \cite{stoc-paper}. In particular, \cite{stoc-paper} introduces the notion of $(K_t,F)$-minor colorability of a pattern $B$, which is a coloring of $B$ with $t$ colors such that the coloring imposes a $t$-clique minor on any copy of $F$ in $B$. Then using this definition, one finds a minimal covering of the graph $H$ with $(K_t,F)$-minor colorable subsets and one argues that one can take one of these subsets as $H'$.

We notice that the properties that \cite{stoc-paper} uses relating the chromatic number and the clique minor of a pattern in this construction can be summarized into the core of patterns. We introduce the notion of $F$-coloring, which simply says that if $B$ is $F$-colorable then there is a coloring such that any copy of $F$ in $B$ is a colorful copy under this coloring. Then we cover $H$ with minimal number of $C$-colorable subsets. We show that we can take one of these subsets as $H'$.

Finally, we generalize Theorem \ref{thm:general-hardness} to the problem of detecting a pattern from a set $S$ of patterns in Theorem \ref{thm:set-hardness}. We show that if $S$ is a set of patterns, there is a pattern $H\in S$, such that detecting the core of $H$, $C$, in a graph $G$ can be reduced to detecting any pattern from $S$ in a graph $G^*$ constructed from $G$ and $S$. In fact, $G^*$ is the reduction graph of Reduction (4) on $H$ as the pattern. The main part of Theorem \ref{thm:set-hardness} is to find the appropriate $H$ in $S$. In order to find this pattern $H$, we look at homomorphisms between the patterns in $S$. In particular, we form a graph with nodes representing patterns in $S$ and directed edges representing homomorphisms. We look at a strongly connected component of this graph that has no edges from other components to it, so there is no homomorphism from any pattern outside this component to any pattern inside the component. We show that all the patterns in this component have the same core and we show that the pattern $H$ can be any of the patterns in this component.






\section{Lower bounds}
\label{sec:lower-bounds}
\subsection{Hardness of PSI}
\label{sec:hardnessofpsi}
In this section we prove Theorem \ref{thm:psihardness}, which reduces a $\eta(H)$-clique detection to $H$-detection, where $\eta(H)$ is the size of the largest clique minor of $H$. 

We can represent a clique minor of $H$ of size $t$ by a function in the following definition. 

\begin{definition}
\label{def:kt-minor-function}
Let $f:V(H)\rightarrow\{1,\ldots,t\}$ be a function such that for any $i\in \{1,\ldots,t\}$, the preimage of $i$, $f^{-1}(i)$, induces a connected subgraph of $H$ and for every $i,j\in \{1,\ldots,t\}$, there is at least one edge between the preimages $f^{-1}(i)$ and $f^{-1}(j)$. We call such $f$ a {\em $K_t$-minor function} of $H$. 
\end{definition}

\begin{figure}
    \centering
    \includegraphics{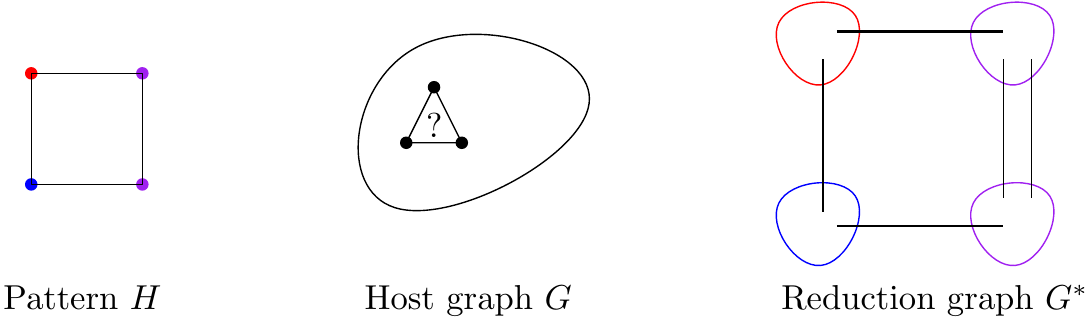}
    \caption{Example of the Reduction of Theorem \ref{thm:psihardness} for pattern $C_4$ with a $K_3$-minor. From left to right: The pattern $H$ with with a $K_3$ minor function shown as a coloring, the host graph $G$ in which we want to find a triangle and the reduction graph $G^*$ built from $G$ and $H$. The bold edges represent the edges in $G$, whereas the double edge represents a perfect matching. Each of the four colored parts in $G^*$ are a copy of $G$.}
    \label{fig:ktminor}
\end{figure}

One can think of $f$ as a coloring on vertices of $H$ that imposes a clique minor on $H$. Figure \ref{fig:ktminor} shows an example of a $K_3$-minor function of $C_4$ as a coloring. In the reduction we are going to consider a $K_t$-minor function $f$ for $t=\eta(H)$. We can find a maximum clique minor of $H$ and its associated function in $O_k(1)$\footnote{Any function that has dependency on $k$ and no other parameter is of $O_k(1)$} as follows: 
Check for all functions $f:V(H)\rightarrow\{1,\ldots,k\}$ if $f$ is a $K_t$-minor function for some $t$, and then take the $f$ that creates a maximum $K_t$-minor. 


\psihardness*

\begin{proof}
Let the size of the maximum clique minor of $H$ be $t$, i.e. $\eta(H)=t$ and let $f:V(H)\rightarrow \{1,\ldots,t\}$ be a $K_t$-minor function of the pattern $H$. Using the function $f$ and the graph $G$, we construct the reduction graph $G^*$ as follows: 

The vertex set of $G^*$ consists of partitions $G_v^*$ for each $v\in V(H)$, where the partition $G_v^*$ is a copy of the vertices of $G$ as an independent set for all $v\in V(H)$ 

The edge set of $G^*$ is defined as follows. For every two vertices $v$ and $u$ in the pattern $H$ where $vu$ is an edge and $f(v)\neq f(u)$, we add the following edges between $G_v^*$ and $G_u^*$:
for each $w_1$ and $w_2$ in $G$, add an edge between the copy of $w_1$ in $G_v$ and the copy of $w_2$ in $G_u$ if and only if $w_1w_2$ is an edge in $G$. In other words, we put the same edges as $E(G)$ between $G_u^*$ and $G_v^*$ in this case.
For any two vertices $v$ and $u$ in $H$ where $vu$ is an edge and $f(v)= f(u)$, add the following edges between $G_v^*$ and $G_u^*$:
for any $w\in G$, connect the two copies of $w$ in $G_v^*$ and $G_u^*$. In other words, we put a complete matching between $G_v^*$ and $G_u^*$ in this case. 
This completes the definition of $G^*$. See Figure \ref{fig:ktminor} for an example. 

Note that $G^*$ is an $H$-partite graph with $nk$ vertices and since for each pair of vertices $u,v\in H$ we have at most $m$ edges between $G_u^*$ and $G_v^*$, the construction time is at most $O(k^2m+kn)\le O(k^2n^2)$.

Now to prove the correctness of the reduction, first we show that the reduction graph $G^*$ has a subgraph isomorphic to $H$ if $G$ has a $t$-clique. Suppose that the vertices $w_1,\ldots,w_t\in V(G)$ form a $t$-clique. Let $H^*$ be the subgraph induced on the following vertices in the reduction graph $G^*$: For each $v\in H$, pick $w_{f(v)}$ from $G_v^*$. We need to show that if $vu\in E(H)$, then there is an edge between the vertices picked from $G_v^*$ and $G_u^*$. This is because if $f(v)=f(u)=i$, then we picked $w_i$ from both $G_v^*$ and $G_u^*$ and hence they are connected. If $f(v)\neq f(u)$, then since $w_{f(v)}$ is connected to $w_{f(u)}$ in $G$, we have that their copies in $G_v^*$ and $G_u^*$ are connected as well. So $H^*$ is isomorphic to $H$. 

Now we show that $G$ has a $t$-clique if $G^*$ has a colorful subgraph $H^*$ isomorphic to $H$. 
Let $v^*\in V(H^*)$ be the vertex picked from $G_v^*$, for $v\in V(H)$. Since there is no edge between $G_v^*$ and $G_u^*$ if $uv$ is not an edge in $H$, we have that there {\em must} be an edge between $v^*$ and $u^*$ if $uv$ is an edge in $H$, so that the number of edges of $H^*$ matches that of $H$. So if $uv\in E(H)$ and $f(u)=f(v)$, then $u^*$ and $v^*$ must be the copies of the same vertex in $G$. Since the vertices with the same value of $f$ are connected, the vertices of $H^*$ are the copies of exactly $t$ vertices in $G$, say $\{w_1,\ldots, w_t\}$, where $v^*$ is the copy of $w_i$ if $f(v)=i$. For each $i,j\in \{1,\ldots,t\}$, there are two vertices $u,v\in V(H)$ such that $f(u)=i$, $f(v)=j$ and $uv\in E(H)$. So $u^*v^*\in E(H^*)$, and hence $w_iw_j\in E(G)$. So the set $\{w_1,\ldots, w_t\}$ induces a $t$-clique in $G$.
\end{proof}
Recall that Corollary \ref{thm:core-hardness} gives a hardness result for cores in SI. This Corollary comes from the result of Marx \cite{marx2007can} that PSI and SI are equivalent when the pattern is a core.

\corehardnessinsi*

We are going to use this result later for proving tighter hardness results for paths and cycles. Now we prove Corollary \ref{cor:coreLB} that gives a lower bound for induced SI when the pattern is a core.

\corehardnessgeneral*


\begin{proof}
To get a lower bound for induced SI when the pattern is a core, we use two results on the connection of the maximum independent set $\alpha(H)$, maximum clique size $w(H)$ and the size of the maximum clique minor $\eta(H)$ of a pattern $H$. Kawarabayashi \cite{kawarabayashi2005improvements} showed $(2\alpha(H)-1)\cdot \eta(H)\ge |V(H)|+w(H)$, and Balogh and Kostochka \cite{balogh2011large} showed that $\alpha(H)\eta(H)\ge |V(H)|/(2-c)$ for a constant $c>1/19.5$. Since $\eta(H)\ge w(H)$, these results imply that $\alpha(H)\eta(H)\ge \max\{(|V(H)|+2w(H))/2, |V(H)|/1.95\}$. Since $\eta(\bar{H})\ge w(\bar{H})=\alpha(H)$ and all these numbers are integers, we get Corollary~\ref{cor:coreLB} from Corollary~\ref{thm:core-hardness}.
\end{proof}

\subsection{Patterns are at least as hard to detect as their core}
\label{sec:corehardness}

In this section we prove that detecting a pattern is at least as hard as detecting its core. In order to do so we define the notions of $C$-coloring and $C$-covering for a core subgraph $C$.
\begin{definition}
Let $F$ be a graph and let $C$ be a $c$-node subgraph of it. We say that the function $f:V(F)\rightarrow \{1,\ldots,c\}$ is a {\em $C$-coloring} of $F$ if for any copy of $C$ in $F$, the vertices of this copy receive distinct colors. We say that a graph is $C$-colorable if it has a $C$-coloring.
\end{definition}

Note that a $C$-coloring of $F$ partitions $F$ into $c$ sections such that any copy of $C$ in $F$ is a colorful copy, i.e. it has exactly one vertex in each partitions\footnote{Note that this is different than $F$ being a $C$-partite graph. The colors are not assigned to any node of $C$, and there is no constraints on the edges of $F$ with respect to the partitions.}. See figure \ref{fig:C-covering} for an example of $C$-coloring for $C$ being the $5$-cycle. 

\begin{figure}
    \centering
    \includegraphics{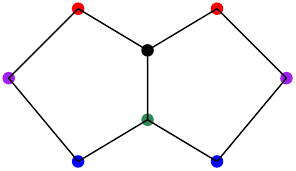}
    \caption{$C$-coloring with $C$ being the $5$-cycle. The pattern has $2$ copies of $C_5$, and each copy is colorful with respect to the coloring given.}
    \label{fig:C-covering}
\end{figure}

\begin{definition}
 Let $H$ be a graph and $C$ be a core of $H$. We say that a collection $\mathcal{C}=\{C_1,\ldots,C_r\}$, $C_i\subseteq V(H)$ is a {\em $C$-covering} for $H$ of size $r$, if the following hold.
 \begin{enumerate}
     \item For every copy of $C$ in $H$ there is an $i$ such that this copy is in the subgraph induced by $C_i$.
     \item For every $i$ the subgraph induced by $C_i$ is $C$-colorable. 
 \end{enumerate}
\end{definition}
For any pattern $H$ with core $C$ there is a simple $C$-covering: Let the sets in the collection be the copies of $C$ in $H$. However, we are interested in the ``smallest" $C$-covering.

\begin{definition}
Define the $C$-covering number of $H$ as the minimim integer $r$ such that there is a $C$-covering for $H$ of size $r$.
\end{definition}
One can find a $C$-covering of minimum size in $O_{k}(1)$ by first enumerating all copies of $C$ in $H$, and then considering all ways of partitioning the copies into sets, and testing if these sets are $C$-colorable. Before proving Theorem \ref{thm:general-hardness}, we prove the following simple but useful lemma.
\begin{lemma}
\label{lem:hom-H-partite}
Let $G$ be an $H$-partite graph where for each $v\in V(H)$, $G_v\subseteq G$ is the partition of $G$ associated to $v$.
Let $F$ be a subgraph in $G$. Then there is a homomorphism $g$ from $F$ to $H$, defined as $g:V(F)\rightarrow V(H)$ where $g(u)=v$ if $u\in G_v$, for every $u\in V(F)$.
\end{lemma}

\begin{proof}
To prove that $g$ is a homomorphism, we need to show that if $v_1v_2\in E(F)$, then $g(v_1)g(v_2)\in E(H)$. This is true because the edge $v_1v_2$ is between $G_{g(v_1)}$ and $G_{g(v_2)}$, and from the definition of $H$-partite graphs this means that $g(v_1)g(v_2)\in E(H)$.
\end{proof}

\generalhardness*

\begin{proof}
We use the color-coding trick of Alon, Yuster and Zwick \cite{alon1995color}: Consider a random assignment of colors $\{1,\ldots,c\}$ to the vertices of the host graph $G$, and a random assignment of numbers $\{1,\ldots,c\}$ to the vertices of $C$. We can assume that if $G$ has a copy of $C$, then the copy of vertex $i$ has color $i$ with high probability (we can repeat this reduction to produce $O(\log{n})$ instances to achieve this high probability). Let the partition $G^{(i)}$ be the vertices with color $i$.

Let the $C$-covering number of the pattern $H$ be $r$, and let $\mathcal{C}=\{C_1,\ldots,C_r\}$ be a $C$-covering of size $r$. Note that as explained before, we can find $r$ and $\mathcal{C}$ in $O_k(1)$ time. 
Let $f:C_1\rightarrow \{1,\ldots,c\}$ be a $C$-coloring of $C_1$, where $c=|V(C)|$ is the size of the core $C$.

We define the vertex set of the $H$-partite  reduction graph $G^*$ by adding a subset of vertices of $G$ for each vertex $v\in C_1$ as the partition associated to $v$, and then simply adding a copy of the rest of the vertices of $H$ to $G^*$. More formally,
for each vertex $v\in C_1$, let $G_v^*$ be a copy of the partition $G^{(f(v))}$ as an independent set. For each vertex $v\in V(H)\setminus C_1$, let $G^*_v=\{v^*\}$ include a copy of $v$ in $G^*$. This finishes the vertex set definition. 

We define the edge set of the reduction graph $G^*$ as follows: For each pair of vertices $u,v\in C_1$, if $uv$ is an edge and $f(u)= f(v)$, then we add a perfect matching between $G_u^*$ and $G_v^*$ as follows: For each $w\in G$, we add an edge between the copy of $w$ in $G_u^*$ and the copy of $w$ in $G_v^*$. If $uv$ is an edge and $f(u)\neq f(v)$, then we add all the edges in $G^{(f(u))} \times G^{(f(v))}$ to $G^*_u\times G^*_v$ as follows: for each $w_1$ and $w_2$ in $G$, we add an edge between the copy of $w_1$ in $G_u^*$ and the copy of $w_2$ in $G_v^*$ if and only if $w_1w_2$ is an edge in $G$. For each pair of vertices $u\in C_1$ and $v\in V(H)\setminus C_1$ such that $uv$ is an edge in $H$, we add an edge between $v^*\in G^*_v$ and all vertices in $G_u^*$. For each pair of vertices $u,v \in V(H)\setminus C_1$ such that $uv$ is an edge in $H$, we add an edge between $u^*\in G_u^*$ and $v^*\in G_v^*$.

Note that the number of edges of $G^*$ is at most $O(mk^2)$ where $m$ is the number of edges of $G$. This is because for every $v\in V(H)\setminus C_1$, the number of edges attached to $G_v^*=\{v^*\}$ is at most $O(nk)$, and for every $u,v\in C_1$, there are at most $m$ edges between $G_u^*$ and $G_v^*$. So the construction time is $O(mk^2)\le O(n^2)$. 

Before proceeding to the proof of the reduction, note that if $uv\notin E(H)$, there is no edge between $G^*_u$ and $G^*_v$. So we have the following observation.
\begin{observation}
\label{obs:G*-H-partite}
$G^*$ is $H$-partite.
\end{observation}

Now we prove that the reduction works. 
First suppose that $G$ has a colorful copy $C'=\{v_1,\ldots,v_c\}$ of $C$, such that $v_i$ has color $i$. We are going to pick $k$ vertices in the reduction graph $G^*$, one from each partition, and prove that they induce a copy of $H$ in $G^*$. For every $v\in C_1$, we pick the copy of $v_i$ in the partition $G_v^*$, where $i$ is the color of $v$ in the $C$-coloring $f$ of $C_1$, i.e. $f(v)=i$. For $u\in V(H)\setminus C_1$ we pick the only vertex in $G_u^*=\{u^*\}$. 

To prove that these $k$ nodes induce a copy of $H$, consider $u,w\in V(H)$ where $uw\in E(H)$. We show that the vertices picked from $G^*_u$ and $G^*_w$ are connected. If one of $u$ and $w$ is not in $C_1$, then all nodes in $G^*_u$ is connected to all nodes in $G^*_w$. If both $u,w$ are in $C_1$, we have two cases. If $u$ and $w$ have the same color, i.e. $f(u)=f(w)=i$, then we have picked copies of $v_i$ from both $G^*_u$ and $G^*_w$, and from the definition of $G^*$ they are connected. If $u$ and $w$ don't have the same color, i.e. $f(u)\neq f(w)$, then we have picked $v_{f(u)}$ from $G^*_u$ and $v_{f(w)}$ from $G^*_w$. Since $v_{f(u)}$ and $v_{f(w)}$ are connected in $G$, from the definition of $G^*$ they are also connected in $G^*$. So the vertices we picked from $G^*$ induce a copy of $H$.

Now we are going to show that if there is a copy of $H$ in the reduction graph $G^*$, then there is a copy of $C$ in $G$.  
For $i\in \{1,\ldots,r\}$, let $S_i=\cup_{v\in C_i} G_v^*$. 
Suppose that $G^*$ has a subgraph $H^*$ isomorphic to $H$. To show that $G$ has a copy of $C$, we prove that $H^*$ has a copy of $C$ with all its vertices in $S_1$, and then we show that this subgraph leads us to a copy of $C$ in $G$. 

 First, consider a copy $C^*$ of $C$ in $H^*$. By observation \ref{obs:G*-H-partite}, we can consider the homomorphism that Lemma \ref{lem:hom-H-partite} defines from $C^*$ to $H$: $u\in V(C^*)\rightarrow v\in H$ if $u\in G^*_v$. Since $C$ is the core, the image of $C^*$ defined by the homomorphism must be isomorphic to $C$. So this copy of $C$ in $H^*$ is mapped to a copy of $C$ in $H$. 

Thus each copy of $C$ in $H^*$ maps to a copy of $C$ in $H$. Note that this copy is in $C_i$ if and only if the copy of $C$ in $H^*$ is in $S_i$.
Now suppose that there is no copy of $C$ in $H^*\cap S_1$. Then each copy of $C$ in $H^*$ is mapped to a copy of $C$ in $H$ that is not in $C_1$, and thus it is in $C_i$ for $i\ge 2$. So the copies of $C$ in $H^*$ are covered by $\mathcal{S}=\{S_2\cap H^*,\ldots, S_r\cap H^*\}$. If we show that for all $i$, $S_i\cap H^*$ is $C$-colorable, then $\mathcal{S}$ is a $C$-covering of size $r-1$ for $H^*$ and since $H^*$ is a copy of $H$, this is a contradiction to the $C$-covering number of $H$. 

To see that $S_i\cap H^*$ is $C$-colorable, let $f_i:C_i\rightarrow \{1,\ldots,c\}$ be the $C$-coloring of $C_i$, for $i=2,\ldots,r$. We color each node $v\in H^*$ as follows. There is $u\in V(H)$ such that $v\in G^*_u$. We color $v$ the same as $u$, with $f_i(u)$. Now we show that each copy $C^*$ of $C$ in $S_i\cap H^*$ has distinct colors. Consider the mapping of Lemma \ref{lem:hom-H-partite} from $C^*$ in the $H$-partite graph $G^*$ to $H$: for $v\in V(C^*)$, we let $g(v)=u$ if $v\in G^*_u$. Note that if $C^*\subseteq H^*\cap S_i$, the map $g$ preserves colors. Since the image of $C^*$ in $H$ is also a copy of $C$ (because $C$ is a core) and $f_i$ is a $C$-coloring, this image is a colorful copy of $C$. So $C^*$ is also a colorful copy of $C$ with the coloring defined. Thus $S_i\cap H^*$ is $C$-colorable.   


So from above we conclude that $H^*$ must have a copy $C^*=\{w_1,\ldots,w_c\}$ of $C$ in $S_1$, such that $w_i\in G_{v_i}$ for some $v_i\in C_1$ and $v_i\neq v_j$ for each $i\neq j$. Moreover, the mapping $w_i\rightarrow v_i$ is a homomorphism from $C^*$ to $H$ and since $C$ is a core, we have that $v_1,\ldots,v_c$ form a a copy of $C$ in $H$. Now since $f$ is a $C$-coloring, $f(v_i)\neq f(v_j)$ for all $i\neq j$. This means that $w_1,\ldots,w_c$ are copies of distinct vertices in $G$, and hence they are attached in $G^*$ if and only if they are attached in $G$. So they form a subgraph isomorphic to $C$ in $G$.
\end{proof}

Now we prove Corollary \ref{cor:induced-si-hardness} and \ref{cor:inducedhardnessone} on induced subgraph isomorphism of \textit{all} patterns. 
\inducedsihardnessone*
\begin{proof}
Denote the chromatic number of a graph $F$ by $X(F)$.
We know that for a $k$ node pattern $H$, the chromatic number of either $H$ or its complement is at least $\sqrt{k}$. WLOG assume that $X(H)\ge \sqrt{k}$. Lemma \ref{lem:color-critical-core} proven below states that  a color critical graph is a core. Since the core of $H$ is its largest subgraph that is a core, we have that $X(core(H))\ge \sqrt{k}$, and so in particular the size of the core of $H$ is at least $\ceil{\sqrt{k}}$. By Theorem \ref{thm:general-hardness} we have that detecting $H$ is at least as hard as detecting $core(H)$, and by Corollary \ref{cor:coreLB} we have that detecting $core(H)$ is at least $\ceil{\sqrt{\sqrt{k}/1.95}}$ hard. This gives the result that we want. 
\end{proof}

\begin{restatable}{corollary}{inducedsihardness} 
\label{cor:induced-si-hardness}
 \textbf{(Hardness of Induced-SI)}
 For any $k$-node pattern $H$, the problem of detecting an induced copy of $H$ in an $n$-node graph requires $n^{\Omega{(\sqrt{k}/\log{k})}}$ time under ETH. 
 \end{restatable}
\begin{proof}
Similar to the proof of Corollary \ref{cor:inducedhardnessone}, we have that $X(core(H))\ge \sqrt{k}$. Now since by inductive coloring we have that for any graph $F$, $tw(F)+1\ge X(F)$, then $tw(core(H))\ge \sqrt{k}-1$. Recall that Marx \cite{marx2007can} shows that under ETH, for any pattern $F$ partitioned subgraph isomorphism of $F$ in an $n$ node graph requires $n^{\Omega(\frac{tw(F)}{\log tw(F)})}$ time. Since for cores PSI and SI are equivalent \cite{marx2007can}, we get Corollary \ref{cor:induced-si-hardness}.
\end{proof}


\begin{lemma}
\label{lem:color-critical-core}
Color critical graphs are cores. 
\end{lemma}
\begin{proof}
Let $H$ be a color critical graph, and suppose that there is a homomorphism $f$ from $H$ to $H'$ where $H'$ is a proper subgraph of $H$. Let $c_{H'}:V(H')\rightarrow \{1,2,\ldots,X(H')\}$ be a coloring of $H'$. Then let $c_H$ be the following coloring for $H$. For each $v\in H'$, color all vertices of $f^{-1}(v)$ the same as $v$. This means that for any $u\in V(H)$, $c_H(u)=c_{H'}(f(u))$. Since $f^{-1}(v)$ is an independent set and $c_{H'}$ is a proper coloring, $c_H^{-1}(i)=f^{-1}(c^{-1}_{H'}(i))$ is an independent set for any color $i$. So $c_H$ is a proper coloring for $H$ of size $X(H')$. This is a contradiction because $H$ is color critical and we have that $X(H')<X(H)$. 
\end{proof}

\subsection{Hardness of Paths and Cycles}\label{sec:paths-and-cycles}
In this section, we prove a stronger lower bound for induced path and cycle detection than what the previous results give us. More precisely, we show that a cycle or path of length $k$ is at least as hard to detect as an induced subgraph as a clique of size roughly $3k/4$. This number comes from the largest clique minor of the complement of paths and cycles. This is formalized in the next lemma which is proved in the appendix.

\begin{lemma}
\label{lem:maxcliqueminor}
Let $H$ be a $k$-node pattern that is the complement of a path or a cycle. Then $\eta(H)=\floor{\frac{k+\omega(H)}{2}}$, where $\omega(H)$ is the size of the maximum clique of $H$. Table \ref{table:maxcliqueminor} shows the value of $\eta(H)$. 

\begin{table}[h!]
\centering
 \begin{tabular}{|c | c| c|} 
 \hline
 number of vertices ($k$)  & $\eta(\bar{C_k})$ & $\eta(\bar{P_k})$ \\ [0.5ex] 
 \hline
 $4t$ & $3t$*&  $3t$*\\ \hline
 $4t+1$ &  $3t$&  $3t+1$*\\ \hline
 $4t+2$ &  $3t+1$ & $3t+1$\\ \hline
 $4t+3$ &  $3t+2$&  $3t+2$\\ \hline
\end{tabular}
\caption{Maximum clique minors. (*) For $t=1$, the maximum minor of $\bar{P_4},\bar{C_4},\bar{P_5}$ is 2,2,3 respectively.}
\label{table:maxcliqueminor}
\end{table}
\end{lemma}

Recall the main result of this section below.
\hardnesscyclepath*

First, we show the easier case of odd cycles which was also mentioned in Section \ref{sec:hardnessofpsi}.
With a simple argument we can show that the complement of an odd cycle is a color critical graph. We prove this in the appendix for completeness.

\begin{lemma}
\label{lem:odd}
The complement of an odd cycle is color-critical. 
\end{lemma}

Lemma \ref{lem:odd} together with Lemma \ref{lem:color-critical-core} show that the complement of an odd cycle is a core. Using Corollary \ref{thm:core-hardness} and Lemma \ref{lem:maxcliqueminor}, we have that detecting a $\bar{C_k}$ for odd $k$ is at least as hard as detecting a $\floor{3k/4}$-clique. Since induced detection of a pattern $H$ is at least as hard as not-necessarily-induced detection of $H$, we have the following Theorem.


\begin{theorem}
\label{cor:odd}
For odd $k$, Induced-$C_k$ detection is at least as hard as $\floor{3k/4}$-clique detection.
\end{theorem}


Now we move to the harder case of even cycles and odd and even paths. We would like to get a hardness as strong as the one offered by Theorem \ref{thm:psihardness} and Corollary \ref{thm:core-hardness}, but we can't use these results directly since paths and even cycles (and their complements) are not cores.

As mentioned in the section \ref{sec:technical-overview}, we are going to use the construction of Theorem \ref{thm:psihardness} and shrink a few partitions of the reduction graph $G^*$, i.e. replacing each of these partitions with a single vertex. The next lemma helps us characterize automorphisms of paths and cycles, and so it helps us find the appropriate partitions of $G^*$ to shrink.

\begin{lemma}
\label{lem:hom-path-cycle}
Any automorphism of paths or cycles that has a \textbf{proper} subset of vertices as its image has the following properties:
\begin{itemize}
    \item Let $C_k=v_1\ldots v_kv_1$ be a $k$-cycle for even $k$. Then any homomorphism from $\bar{C_k}$ to a proper subgraph of $\bar{C_k}$ has two vertices both being mapped to either $v_1$ or $v_k$.
    \item Let $P_k=v_1\ldots v_k$ be a $k$-path. Then any homomorphism from $\bar{P_k}$ to a proper subgraph of $\bar{P_k}$ has two vertices both being mapped to either $v_1$ or $v_k$.
\end{itemize}
\end{lemma}
\begin{proof}
We first consider even cycles, then odd paths and finally even paths.

First consider the pattern $\bar{C_k}$ with an automorphism to a proper subset of it, for even $k$. This graph has exactly two $k/2$-cliques: $K_1$ with $V(K_1)=\{v_1,v_3,\ldots,v_{k-1}\}$ and $K_2$ with $V(K_2)=\{v_2,v_4,\ldots,v_k\}$ (this can be seen by the fact that no two vertices of a $k/2$ clique in $\bar{C_k}$ can be adjacent in $C_k$). Since the only automorphism of a clique is a clique, $K_1$ and $K_2$ should be mapped to $K_1$ or $K_2$. Since this automorphism of $\bar{C_k}$ is to a proper subset of it, both $K_1$ and $K_2$ are mapped to $K_1$, or both of them are mapped to $K_2$. In either case, two vertices of $\bar{C_k}$ are mapped to either $v_1$ or $v_k$.

Now consider $\bar{P_k}$ where $P_k=v_1\ldots v_k$ is an odd path. The graph $\bar{P_k}$ has exactly one $\frac{k+1}{2}$-clique $K_1$ with $V(K_1)=\{v_1,v_3,\ldots,v_k\}$. So this clique should be mapped to itself. Now the rest of the graph is a $\frac{k-1}{2}$-clique $K_2$ with $V(K_2)=\{v_2,v_4,\ldots,v_{k-1}\}$. There are a lot of $(\frac{k-1}{2})$-cliques in $\bar{P_k}$ that $K_2$ can be mapped to, however all of them  contain either $v_1$ or $v_k$. This is because $\bar{P_k}\setminus\{v_1,v_k\}$ has exactly one $(\frac{k-1}{2})$-clique which is $K_2$, and since this automorphism is to a proper subset of $\bar{P_k}$, $K_2$ cannot be mapped to itself. So it is mapped to a $(\frac{k-1}{2})$-clique that has either $v_1$ or $v_k$ in its vertex set, and since $K_1$ is mapped to itself, there are two vertices that are both mapped to either $v_1$ or $v_k$.
    
Finally, consider the complement of an even path $\bar{P_k}$ as our graph. Consider these two $k/2$-cliques $K_1$ and $K_2$ in this graph: $V(K_1)=\{v_1,v_3,\ldots,v_{k-1}\}$ and $V(K_2)=\{v_2,v_4,\ldots, v_k\}$. First we observe that the only $k/2$-clique in $\bar{P_k}\setminus v_k$ is $K_1$ and the only $k/2$-clique in $\bar{P_k}\setminus v_1$ is $K_2$. Moreover, $\bar{P_k}\setminus \{v_1,v_k\}$ does not have any $k/2$-cliques. So the mappings of $K_1$ and $K_2$ must use at least one of $v_1$ or $v_k$. If none of $v_1$ and $v_k$ have two vertices mapped to them, then it must be that the mapping of $K_1$ is using exactly one vertex in $\{v_1,v_k\}$, and so by the observation above it must be mapped to either $K_1$ or $K_2$. The same goes for $K_2$. But since this automorphism of $\bar{P_k}$ is to a proper subset of the vertices, it must be that both $K_1$ and $K_2$ are mapped to either $K_1$ or $K_2$. So there are two vertices both mapped to either $v_1$ or $v_k$
\end{proof}

Now we are ready to prove Theorem \ref{thm:hardnesscyclepath}. 

\begin{proofof}{Theorem \ref{thm:hardnesscyclepath}}
The idea is to make a small change to the construction of Theorem \ref{thm:psihardness} so that if the pattern $H$ has a copy in $G^*$, then all vertex sets $G_v^*$s have exactly one vertex of the copy. We explain the construction of Theorem \ref{thm:psihardness} here again for completeness. If $\bar{H}$ is a path, let its vertices be in the order $v_1\ldots v_k$, and if it is an even cycle, let the cycle be $v_1\ldots v_kv_1$. We prove a slightly more general statement. We show that for $t'\in \{t-1,t\}$ which depends on the $K_t$ minor function of $H$, detecting $H$ is at least as hard as detecting a $t'$-clique. Since detecting $(t-2)$-clique reduces to detecting $(t-1)$-clique, this proves the theorem.
Given a graph $G$ in which we want to find a $t'$-clique we construct the $H$-partite reduction graph $G^*$ of Theorem \ref{thm:psihardness} as follows.
 
Let $f:V(H)\rightarrow \{1,\ldots,t\}$ be a $K_t$-minor function of the pattern $H$ (see Definition \ref{def:kt-minor-function}). Let $t' = t-1$ if $f(v_1)=f(v_k)$, and let $t'=t-2$ otherwise. We can also assume that $f(v_1)=t'+1$ and $f(v_k)=t$. Construct the reduction graph $G^*$ as follows: For each vertex $v_i\in H$, let $G_{v_i}^*$ be a copy of the vertices of $G$ as an independent set. For every two vertices $v_i$ and $v_j$ in the pattern $H$ where $v_iv_j$ is an edge and $f(v_i)\neq f(v_j)$, add the following edges between partitions $G_{v_i}^*$ and $G_{v_j}^*$:
For each $w_1$ and $w_2$ in $G$, add an edge between the copy of $w_1$ in partition $G_{v_i}^*$ and the copy of $w_2$ in partition $G_{v_j}^*$ if and only if $w_1w_2$ is an edge in $G$. 
For any two vertices $v_i$ and $v_j$ in $H$ where $v_iv_j$ is an edge and $f(v_i)= f(v_j)$, add the following edges between partitions $G_{v_i}^*$ and $G_{v_j}^*$:
For every $w\in G$, connect the copy of $w$ in partition $G_{v_i}^*$ and the copy of $w$ in partition $G_{v_j}^*$. This completes the definition of the reduction graph $G^*$, before we make modifications to it.
Note that $G^*$ is an $H$-partite graph with $nk$ vertices and since for each pair of vertices $v_i,v_j\in H$ we have at most $m=O(n^2)$ edges between $G_{v_i}^*$ and $G_{v_j}^*$, the construction time is at most $O(k^2n^2)$.
 
Now we do the following modifications to this construction: For each $v_i$ with $f(v_i)\in \{f(v_1),f(v_k)\}$, remove the vertices in the sets $G_{v_i}^*$ and add a single vertex instead so that $G_{v_i}^*=\{v_i^*\}$. Note that we also have $G_{v_1}^*=\{v_1^*\}$ and $G_{v_k}^*=\{v_k^*\}$. For $j=1,\ldots,k$, add edges between $v_i^*$ and all vertices of $G_{v_j}^*$ if there is an edge between $v_i$ and $v_j$ in $H$.
Note that after these modifications, $G^*$ stays $H$-partite.

Now we show that the reduction works.  
First suppose that $G$ has a $t'$-clique $\{u_1,\ldots,u_{t'}\}$. We pick $k$ vertices in the reduction graph $G^*$ and show that they form a copy of $H$. For each $v_i\in V(H)$ where $f(v_i)\notin \{f(v_1),f(v_k)\}$, pick the copy of $u_{f(v_i)}$ from partition $G_{v_i}^*$. Recall that $f(v_i)\le t'$ in this case. 
For each $v_i\in V(H)$ where $f(v_i)\in \{f(v_1),f(v_k)\}$, pick $v_i^*$ from $G^*_{v_i}$, so that we have $k$ vertices in total.

We need to show that if $v_iv_j$ is an edge in $H$, then there is an edge between the vertices picked from $G_{v_i}^*$ and $G_{v_j}^*$. If $f(v_i)\in \{f(v_1),f(v_k)\}$, then since $v_i^*$ is attached to all nodes in $G_{v_j}^*$ if $v_iv_j\in E(H)$, then $v_i^*$ is attached to the vertex chosen from $G_{v_j}^*$. So assume that $f(v_i),f(v_j)\notin \{f(v_1),f(v_k)\}$. If $f(v_i)=f(v_j)=\ell$, then we picked copies of $u_{\ell}$ from both partitions $G_{v_i}^*$ and $G_{v_j}^*$ and so they are connected. If $f(v_i)\neq f(v_j)$, then since $u_{f(v_i)}$ is connected to $u_{f(v_j)}$ in $G$, we have that their copies in $G_{v_i}^*$ and $G_{v_j}^*$ are connected in $G^*$ as well. So $H^*$ is isomorphic to $H$.   
 
Now if $G^*$ has a copy $H^*$ of $H$, let $g$ be the following function from $H^*$ to $H$: for every $w\in V(H^*)$, let $g(w)=v$ if $w\in G_v^*$. In fact $g$ is a homomorphism by the way $G^*$ is constructed. The image of $g$ is the set $A=\{v\in V(H)| \exists w\in V(H^*) \text{ such that } w\in G_v^*\}$. If $A$ is a proper subset of $H$, then by Lemma \ref{lem:hom-path-cycle}, $g$ maps two vertices of $H^*$ to either $v_1$ or $v_k$. This means that $H^*$ has two vertices in either $G_{v_1}^*$ or $G_{v_k}^*$. However, these sets only have one vertex. So the subgraph that $A$ induces in $H$ is a copy of $H$ and $H^*$ has exactly one vertex $w_i$ in each $G_{v_i}^*$. We show that the vertices $w_i$ help us find a $t'$-clique in $G$.
 
First note that for each $i,j$ with $v_iv_j\in E(H)$, we have $w_iw_j\in E(G^*)$. This is because there is no edge between $G_{v_i}^*$ and $G_{v_j}^*$ if $v_iv_j$ is not an edge in $H$, and so there {\em must} be an edge between $w_i$ and $w_j$ if $v_iv_j$ is an edge in $H$, so that the number of edges of $H^*$ matches that of $H$. Consider the $K_t$-minor function $f:V(H)\rightarrow \{1,\ldots, t\}$ of $H$ with which the reduction graph $G^*$ is constructed. Recall that $f(v_1),f(v_k)\in \{t'+1,t\}$. Now for each $i,j$ where $f(v_i)=f(v_j)\notin \{f(v_1),f(v_k)\}$ and $v_iv_j\in E(H)$, we have that $w_iw_j\in E(H^*)$ and so $w_i$ and $w_j$ are copies of the same vertex in $G$. So for each $\ell=1,\ldots,t'$, since the preimage $f(\ell)^{-1}$ is connected, we have that there is some $z_{\ell}\in G$ such that for all $v_i\in f(\ell)^{-1}$, $w_{i}$ is a copy of $z_{\ell}$. For each $\ell,\ell'\in \{1,\ldots,t'\}$, there are $v_i\in f(\ell)^{-1}$ and $v_j\in f(\ell')^{-1}$ such that $v_iv_j\in E(H)$. So $w_iw_{j}\in E(G^*)$, and thus $z_{\ell}z_{\ell'}\in E(G)$. So $z_1,\ldots,z_{t'}$ form a $t'$-clique.
\end{proofof}

Using Lemma \ref{lem:maxcliqueminor} and Theorem \ref{thm:hardnesscyclepath} we have the following Corollary.

\begin{corollary}
Detecting $C_k$ or $P_k$ as an induced subgraph is at least as hard as detecting a $(\floor{3k/4}-2)$-clique.
\end{corollary}

\section{Paired Pattern Detection}
\label{sec:paired}
In this section we look at hardness results as well as algorithms for Paired Pattern Detection (and more generally Subset Pattern Detection). Recall that for a set of patterns $S$, by (induced) $S$-detection we mean finding a (induced) copy of one of the patterns in $S$, or indicating that there is no copy of any of the patterns in $S$.
\subsection{Hardness for Subset Pattern Detection}

Suppose that we want to prove hardness for detecting $S=\{H_1,\ldots, H_s\}$ in a host graph $G$. Let $C_i$ be the core of $H_i$ for all $i=1,\ldots,s$. We first prove a series of lemmas about the relations between the cores $C_i$, and then use these lemmas to prove Theorem \ref{thm:set-hardness} below whose corollary was stated in the introduction.


\begin{restatable}{theorem}{sethardness}
\label{thm:set-hardness}
Let $G$ be an $n$-node host graph and let $S$ be a set of patterns. There is a pattern $H\in S$, with core $C$ such that one can construct an $O(n)$-node graph $G^*$ in $O(n^2|V(H)|^2)$ time where $G$ has a copy of $C$ as a subgraph if and only if $G^*$ has a copy of a pattern in $S$ as a subgraph.
\end{restatable}

Let $G_i^*$ be the reduction graph of Theorem \ref{thm:general-hardness} for pattern $H_i$: detecting $C_i$ in a host graph $G$ reduces to detecting $H_i$ in $G_i^*$. We are going to state a few lemmas that help us prove Theorem \ref{thm:set-hardness}.

\begin{lemma}
\label{lem:not-hom}
Let $H_1$ and $H_2$ be two patterns. If there is no homomorphism from $H_1$ to $H_2$, then there is no copy of $H_1$ in $G_2^*$. 
\end{lemma}
\begin{proof}
Since by Observation \ref{obs:G*-H-partite} $G_2^*$ is a $H_2$-partite graph, if there is a copy of $H_1$ in $G_2^*$, then by Lemma \ref{lem:hom-H-partite} $H_1$ has a homomorphism to $H_2$, which is a contradiction. 
\end{proof}
\begin{corollary}
Let $S=\{H_1,\ldots, H_s\}$ be a set of patterns such that $C_j$ is the core of $H_j$ for all $j=1,\ldots,s$. Then if there is $H_i$ such that for all $j\neq i$ there is no homomorphism from $H_j$ to $H_i$, there is no copy of $H_j$ in $G_i^*$ for all $j\neq i$.
\end{corollary}

\begin{lemma}
\label{lem:c-covering-num}
Let $H_1,H_2$ be two patterns with isomorphic cores $C$ such that there is a homomorphism from $H_1$ to $H_2$. Then the $C$-covering number of $H_1$ is at most as big as the $C$-covering number of $H_2$. 
\end{lemma}
\begin{proof}
Let $g:V(H_1)\rightarrow V(H_2)$ be a homomorphism from $H_1$ to $H_2$. First note that from the definition of core, $g$ takes any copy of $C$ in $H_1$ to a copy of $C$ in $H_2$. Let $\mathcal{C}=\{C_1,\ldots,C_r\}$ be a $C$-covering of $H_2$. Then $\{g^{-1}(C_1),\ldots, g^{-1}(C_r)\}$ is a $C$-covering of $H_1$: First, suppose $C'$ is a copy of $C$ in $H_1$, and suppose that $g$ takes $C'$ to $g(C')$ which is a copy of $C$ in $H_2$ (copies of $C$ in $H_1$ must be take to copies of $C$ in $H_2$ since $C$ is a core). So there is $1\le i\le r$ such that $g(C')\subseteq C_i$. Hence $C'\subseteq g^{-1}(C_i)$.

Now we show that $g^{-1}(C_i)$ is $C$-colorable for all $i$. Note that $g^{-1}(v)$ is an independent set for any $v\in V(H_2)$. So if $f_i:V(H_2)\rightarrow \{1,\ldots,|C|\}$ is a $C$-coloring for $C_i$, then color $u\in g^{-1}(C_i)$ with $f(g(u))$. Now if $C'$ is a copy of $C$ in $H_1$, all vertices of $C'$ have different colors. This is because all vertices of $g(C')$ have different colors by the definition of $f$.

So we found a $C$-covering of $H_1$ with $r$ sets, where $r$ is the $C$-covering number of $H_2$. Since the $C$-covering number is the size of the smallest $C$-covering, we proved the lemma. 
\end{proof}

\begin{lemma}
\label{lem:equalcores}
Let $S=\{H_1,\ldots, H_s\}$ be a set of patterns such that $C_j$ is the core of $H_j$ for all $j=1,\ldots,s$. Suppose that there is a homomorphism from $H_{j}$ to $H_{j+1}$ for all $j=1,\ldots,s$, where $j$ is taken mod $s$. Then all $C_i$s are isomorphic. Moreover, if all $C_i$s are isomorphic to $C$, the $C$-covering number of all of the patterns is the same. 
\end{lemma}
\begin{proof}
Note that since $H_j$ has a homomorphism to $H_{j+1}$, $C_j$ is a subgraph of $H_j$ and $H_{j+1}$ has a homomorphism to $C_{j+1}$, we have that $C_j$ has a homomorphism to $C_{j+1}$. 
WLOG suppose that $C_1$ has the highest number of edges among  $C_j$s for $j=1,\ldots,s$. Let $g_j:V(C_j)\rightarrow V(C_{j+1})$ be a homomorphism from $C_j$ to $C_{j+1}$, for all $j=1,\ldots, s$. So $g=g_1 \circ \ldots \circ g_s$ is a homomorphism from $C_1$ to $C_1$. Since $C_1$ is a core, the image $g(C_1)$ must be isomorphic to $C_1$. This in particular means that the size of the image $g_1(C_1)$ is the same as $C_1$ and $g_1$ in an injection. So no two edges are mapped to one edge, and so $|E(C_2)|\ge |E(C_1)|$. So the number of edges of $C_1$ and $C_2$ is the same and all the edges of $C_2$ are in the image of $g_1$. Moreover, since $C_2$ is a core, it has no single vertex. So all nodes of $C_2$ are in the image of $g_1$. Hence, $C_2$ and $C_1$ are isomorphic. Similarly we can reason about $g_2$ next, and we can say that all cores are isomorphic. 

Now by Lemma \ref{lem:c-covering-num} the $C$-covering number of all $H_i$s are equal. 
\end{proof}

\begin{lemma}
\label{lem:equal-cores-reduction}
Let $S=\{H_1,\ldots, H_s\}$ be a set of patterns and let $C$ be a core graph and suppose that $C$ is the core of all patterns in $S$. Let $G$ be a host graph, and suppose that $H_{1}$ has the minimum $C$-Covering number among all patterns in $S$ (there might be other patterns in $S$ with the same $C$-covering number). Then if there is a copy of $H_i$ in $G_1^*$, $H_i$ must have the same $C$-covering number as $H_1$, and there is a copy of $C$ in $G$ with high probability.
\end{lemma}
\begin{proof}
Since we are going to use the reduction graph $G_1^*$, for the sake of completeness we are going to explain the construction of this graph. We denote $G_1^*$ by $G^*$ for simplicity. Let $\mathcal{C}=\{C_1,\ldots,C_{r}\}$ be a minimum $C$-covering of $H_1$, where $n_1=r$ is the $C$-covering number of $H_1$. We color-code $G$ with $|C|=c$ colors $\{1,\ldots,c\}$ and assume that if $G$ has a copy of $C$ in it, it is a colorful copy with high probability. 

The reduction graph $G^*$ is a $H_1$-partite graph, with partitions $G^*_v$ for $v\in H_1$, such that if $v\notin C_1$ then the partition $G_v^*$ has only a copy of $v$ and if $v\in C_1$, the partition $G_v^*$ has a copy of a subset of $G$. Particularly if $f_1:C_1\rightarrow \{1,\ldots,c\}$ is a $C$-coloring of the subgraph $C_1\subseteq H$, then the partition $G_v^*$ is the set of nodes in $G$ with color $f_1(v)$. To define the edges of $G^*$, for any $uv\in E(H)$, if one of $u$ and $v$ is not in $C_1$, all the nodes in $G_u^*$ is attached to all the nodes in $G_v^*$. If both $u$ and $v$ are in $C_1$, then if they are of the same color ($f_1(v)=f_1(u)$) we put a complete matching between $G_u^*$ and $G_v^*$: for any $w\in V(G)$ with color $f_1(u)=f_1(v)$, we put an edge between the copies of $w$ in partitions $G_u^*$ and $G_v^*$. If $u$ and $v$ are of different colors, we connect a node in $G_u^*$ to a node in $G_v^*$ if and only if they are connected in $G$. 

Let the $C$-covering number of $H_i$ be $n_i$. By Lemma \ref{lem:hom-H-partite}, there is a homomorphism from $H_i$ to $H_1$, and by Lemma \ref{lem:c-covering-num}, $n_1\ge n_i$. By minimality of $n_1$, we have that $n_1=n_i=r$. 

Now we prove that there is a copy of $C$ in $G$. The proof is similar to Theorem \ref{thm:general-hardness} (but due to small technicalities we can't use Theorem \ref{thm:general-hardness} directly). Let the copy of $H_i$ in $G^*$ be $H_i^*$. By Lemma \ref{lem:hom-H-partite} and the fact that $C$ is a core, any copy of $C$ in $H_i^*$ maps to a copy of $C$ in $H_1$. Let $S_j=\cup_{v\in C_j}G_v^*$ for $j=1,\ldots,r$. We also know that if there is a copy of $C$ in $H_i^*$ that is in $S_j$, then this copy is mapped to a copy of $C$ in $C_j$. 

Now suppose that there is no copy of $C$ in $H^*_i\cap S_1$. Then each copy of $C$ in $H_i^*$ is mapped to a copy of $C$ in $H_1$ that is not in $C_1$. So the copies of $C$ in $H^*_i$ are covered by $\mathcal{S}=\{S_2\cap H^*_i,\ldots, S_r\cap H^*_i\}$. If we show that $S_j\cap H^*_i$ is $C$-colorable for all $j$, then $\mathcal{S}$ is a $C$-covering of size $r-1$ for $H^*_i$ and since $H^*_i$ is a copy of $H_i$, this is a contradiction to the $C$-covering number of $H_i$. 

To see that $S_j\cap H^*_i$ is $C$-colorable, let $f_j:C_j\rightarrow \{1,\ldots,c\}$ be a $C$-coloring of $C_j$. For $v\in S_j\cap H_i^*$, color $v$ the same as $f_j(w)$ if $v\in G_w^*$. Now we see that each copy of $C$ in $S_j$ has distinct colors because it is mapped to a copy of $C$ in $C_j$ with the coloring preserved by the mapping. Since $f_j$ is a $C$-coloring, this copy has distinct colors.

So $H^*_i$ has a copy $C^*=\{w_1,\ldots,w_c\}$ of $C$ in $S_1$, such that $w_i\in G_{v_i}$ for some $v_i\in C_1$ and $v_i\neq v_j$ for each $i\neq j$. Moreover, by Lemma \ref{lem:hom-H-partite}, we have that $v_1,\ldots,v_c$ form a subgraph isomorphic to $C$ in $H_1$. So these nodes must have different colors with respect to the coloring $f_1$, so $f_1(v_i)\neq f_1(v_j)$ for all $i\neq j$. This means that $w_1,\ldots,w_c$ are copies of distinct vertices in $G$, and hence they are attached in $G^*$ if and only if they are attached in $G$. So they form a subgraph isomorphic to $C$ in $G$.
\end{proof}

\begin{proofof}{Theorem \ref{thm:set-hardness}}
To help us find the pattern $H$, we create a directed graph $F$ as follows. The vertices of $F$ are patterns in $S$, and we add an edge from the vertex assigned to $H\in S$ to $H'\in S$ if there is a homomorphism from $H$ to $H'$. Consider the strongly connected components of $F$. These components form a DAG. Consider the strongly connected component $F_1$ that doesn't have any incoming edge from other components to it. Note that by Lemma \ref{lem:equalcores} all patterns in $F_1$ have isomorphic cores, since any two patterns in $F_1$ are in a cycle. Let this shared core be $C$, and let $H\in F_1$ be a pattern with the minimum $C$-covering number among all patterns in $F_1$. 

We show that detecting $C$ in $G$ reduces to detecting $S$ in $G^*$, where $G^*$ is the graph created in Theorem \ref{thm:general-hardness} for pattern $H$, core $C$ and host graph $G$. First, suppose that there is a copy of $C$ in $G$. Then by the proof of Theorem \ref{thm:general-hardness}, there is a copy of $H$ in $G^*$. 

Now suppose that there is a copy of $H'\in S$ in $G^*$. If $H'\notin F_1$, there is no homomorphism from $H'$ to $H$.  By Lemma \ref{lem:not-hom}, this is a contradiction. So $H'\in F_1$. By Lemma \ref{lem:equal-cores-reduction}, there is a copy of $C$ in $G$. 
\end{proofof}

Assume all patterns in $S$ have a $k$-clique. Since there is a homomorphism from a pattern to its core, the core of any pattern must have a $k$-clique. So detecting the core of any pattern in $S$ is at least as hard to detect as a $k$-clique. Thus we obtain Theorem \ref{thm:set-detection-cor} from Theorem \ref{thm:set-hardness}. 

\setdetectioncore*

\subsection{Algorithms}
In this section we focus on algorithms for induced Pair Pattern Detection. Since we are only working on induced detection, we might refer to induced $S$-detection as $S$-detection for any pair $S$. First we give algorithms for detecting sets of $3$-node patterns which are proved in the Appendix. We show that we can detect any pair of $3$-node patterns in $O(n^2)$ time. Next we show that we can detect any pair of $4$-node patterns in $O(n^{\omega})$ time. Afterwards we focus on specific pairs, in particular the case where one of the patterns is $C_4$, and try to decrease this running time. 


\begin{theorem}
\label{thm:3-node-pairs}
Let $H_1$ and $H_2$ be two $3$-node patterns. If $\{H_1,H_2\}\neq \{K_3,P_2\cup I_1\}$ and $\{H_1,H_2\}\neq \{I_3,P_3\}$, then there is an algorithm for induced $\{H_1,H_2\}$ detection that runs in $O(m+n)$ time in an $m$-edge $n$-node host graph. For the cases $\{K_3,P_2\cup I_1\}$ and $\{I_3,P_3\}$, there is an algorithm running in $O(n^2)$ time.
\end{theorem}

\begin{theorem}
Let $H_1$ and $H_2$ be $4$-node patterns. There is a (randomized) algorithm that detects induced $\{H_1,H_2\}$ in an $n$-node graph in $\tilde{O}(n^{\omega})$ time.
\end{theorem}
\begin{proof}
If $n\le 31$, then we check all subgraphs of size $4$ in the host graph to see if they are isomorphic to $H_1$ or $H_2$. So suppose that $n>31$.

First we show how to detect whether a host graph $G$ has an induced copy of $H_1$ or $H_2$. Then using a standard self-reduction technique, if $H_1$ or $H_2$ exist in $G$, we can find a copy of them. An informal description of this approach is the following: We divide the graph into $5$ sections of size roughly $n/5$, and run the detection algorithm on the union of every $4$ sections. If one of these runs outputs YES (that there is a $H_1$ or $H_2$), we recurse on this subgraph of size roughly $4n/5$. Note that if $G$ has a $H_1$ or $H_2$, one of these subgraphs must contain $H_1$ or $H_2$. Our recursion depth is $O(\log{n})$ and so we can find a pattern in the same running time as the detection algorithm with a $O(\log{n})$ overhead\footnote{see \cite{stoc-paper} Section $3$ for a more formal explanation.}. 

Now we give the detection algorithm. By \cite{williams2014finding}, every $4$-node pattern that is not $K_4$ or $I_4$ can be detected in $O(n^{\omega})$ time. So if $\{H_1,H_2\}\cap \{I_4,K_4\}=\emptyset$, then we run the detection algorithms for $H_1$ and $H_2$. So assume that one of the patterns is $I_4$ or $K_4$. 

If $\{H_1,H_2\}=\{K_4,I_4\}$, then since $n>31$, in Lemma \ref{lem:ramsey} we show that we can detect induced $\{H_1,H_2\}$ in linear time. So WLOG suppose $H_1\in \{K_4,I_4\}$ and $H_2\notin \{K_4,I_4\}$.

Let $F\neq I_4$ be a $4$-node pattern and let $g(F)=n_F\cdot aut(F)$ where $n_F$ is the number of occurrences of $F$ in the host graph $G$, and $aut(F)$ is the number of automorphisms of $F$. From \cite{stoc-paper} we know that for any edge $e$ in F, we can compute $g(F) + g(F-e)$ in $O(n^{\omega})$ time.

Now if $H_1$ is an arbitrary graph (with at least 2 edges) and $H_2=H_1-e_1-e_2$ where $e_1$ and $e_2$ are two edges in $H_1$, we can compute $g(H_1)+g(H_1-e1)$ and $g(H_1-e1) - g(H_1-e1-e2)$, and by subtracting these values we get $g(H_1)-g(H_2)$.

In general, for any two arbitrary graphs $H_1$ and $H_2$, we can compute $g(H_1)+ (-1)^r g(H_2)$ where $r$ is a function of the number of edges of $H_1$ and $H_2$. We can do this by considering a set of graphs $H_1=F_0, F_1,\ldots, F_t=H_2$, where for each $i$, $F_i$ and $F_{i+1}$ only differ in one edge, and hence we can compute $g(F_i)+g(F_{i+1})$. Then by combining these values by adding or subtracting each one, we can compute $g(H_1)+ (-1)^r g(H_2)$. 

This means that for any two $4$-node patterns $H_1$ and $H_2$, we can compute the quantity $Q_{H_1,H_2}:=aut(H_1)n_{H_1}+(-1)^{r}aut(H_2)n_{H_2}$ in $O(n^{\omega})$ time for some $r$ that is dependent on $H_1$ and $H_2$. Now since $H_2\notin \{K_4,I_4\}$ we first run the $H_2$ detection algorithm. If it outputs YES we are done. If it outputs NO, then we know that $n_{H_2}=0$, so $Q_{H_1,H_2}= aut(H_1)n_{H_1}$, and so it is non-zero if and only if $G$ has a copy of $H_1$.
\end{proof}

We prove the following two theorems in the next subsections. The first Theorem is also proven in \cite{EschenHSS11}, but we include our proof for completeness.
\begin{theorem}
\label{thm:diamond}
Let $G$ be an $n$-node host graph. Let $H$ be the diamond. Then there is an algorithm for induced detection of $\{C_4,H\}$ in $G$ that runs in $O(n^{7/3})$ time. 
\end{theorem}

\begin{theorem}
\label{thm:clique}
Let $G$ be an $n$-node host graph. Then there is an algorithm for induced detection of $\{C_4,K_4\}$ in $G$ that runs in $O(n^{7/3})$ time. 
\end{theorem}

Patterns of size $4$ that contain a triangle are the clique, diamond, the paw and co-claw. We show in the appendix that there is a $O(n^2)$ algorithms for induced detection of $\{C_4,co-claw\}$ and $\{C_4,paw\}$. This together with Theorem \ref{thm:diamond} and \ref{thm:clique} prove Theorem \ref{thm:c4-general}.

\cfourgeneral*


\subsubsection{Proof of Theorem \ref{thm:diamond}: Detecting $\{diamond,C_4\}$ in $O(n^{7/3})$ time.}
We first prove the following useful lemmas.
\begin{lemma}
\label{lem:p3}
If a graph doesn't have an induced $P_3$, then it is a disjoint union of cliques. If it doesn't have an induced $\bar{P_3}$, then it is a complete $t$-partite graph for some $t$. Additionally, we can detect a $P_3$ ($\bar{P_3}$) or determine that the graph doesn't have a $P_3$ ($\bar{P_3}$) in $O(n^2)$ time. 
\end{lemma}

\begin{proof}
We prove the lemma for $P_3$. The proof for $\bar{P_3}$ is similar: Take the complement of the host graph and search for $P_3$. 

To prove the lemma for $P_3$, take a vertex with maximum degree, say $v$, and consider the set of its neighbors $N(v)$, and let $N[v]=N(v)\cup \{v\}$. Scan all pairs in $N(v)$, if there are two nodes without an edge between them, we have a $P_3$. Otherwise, $N[v]$ is a complete graph. Since $v$ has the maximum degree, no vertex in $N[v]$ is attached to a vertex outside $N[v]$. So this clique is disconnected from the rest of the graph, and we have spent $O(|N[v]|^2)$ time. We do the same procedure for the rest of the graph. By induction, we spend $|N[v]|^2+(n-|N[v]|)^2\le n^2$. If we don't find a $P_3$, $G\setminus N[v]$ is a collection of disjoint cliques by induction, and so is $G$. 
\end{proof}
\begin{lemma}
\label{lem:triangle-and-c4}
There is an algorithm for induced detection of $\{C_4,K_3\}$ that runs in $O(n^2)$ time. Moreover, a $\{C_4,K_3\}$-free graph has at most $O(n^{1.5})$ edges. 
\end{lemma}
\begin{proof}
Run the non-induced $4$-cycle detection algorithm of Richards and Liestman \cite{non-induced-4cycle} which takes $O(n^2)$ time. If it outputs Yes, then we either have an induced $4$-cycle, or a non-induced diamond, in either case we have a triangle or an induced $4$-cycle. 

So suppose that it outputs No. This means that the number of nodes with degree at least $3\sqrt{n}$ is less than $\sqrt{n}$: For the sake of contradiction, let $v_1,\ldots, v_{\sqrt{n}}$ be some of the nodes with degree at least $3\sqrt{n}$. We know that the graph doesn't have a $4$-cycle, so any two nodes have at most $1$ neighbor in common. So each $v_i$ has at least $2\sqrt{n}$ neighbors that are not attached to any $v_j$ for $j\neq i$. This means that the graph has at least $2n$ nodes, a contradiction. So the number of these high degree nodes is at most $\sqrt{n}$, and the graph has $O(n^{1.5})$ edges.  

We need to check if the graph has a triangle, and we do it as follows: for each edge, check in $O(\sqrt{n})$ time if both of its endpoints are attached to any of the high degree nodes. Then for each low degree node, go through every pair of its neighbors and check if they are connected. Since we have $O(n)$ pair of neighbors, this takes $O(n^2)$ time in total.
\end{proof}
Now we prove Theorem \ref{thm:diamond}.
Starting from two nodes that are not attached, we can find a maximal independent set in $O(nt)$ time where $t$ is the size of this set: $I=\{v_1,\ldots,v_t\}$. 
\paragraph{step 1} Let $N(v_j)=N_j$, and note that $\cup_{j=1}^t N_j=V(G)\setminus I$ because of the maximality of $I$. For each vertex $v_i\in I$, we first check if $v_i$ has at most one common neighbor with any of $v_1,\ldots,v_{i-1}$. We can do this in $O(\sum_i deg(v_i))=O(n^2)$ in total. Suppose some $v_i$ violates this: so there is $j< i,$ such that $N_i\cap N_j$ has size at least 2. So $v_i,v_j$ and two of the nodes in $N_i\cap N_j$ form an induced $4$-cycle or diamond. So $|N_i\cap N_j|\le 1$ for all $1\le i<j\le t$. Now using lemma \ref{lem:p3} we check if the subgraph induced on $N_i$ has an induced $P_3$ in $O(|N(v_i)|^2)$ time. If some $N_i$ has an induced $P_3$, then this $P_3$ with $v_i$ forms a diamond. If $N_i$ doesn't have an induced $P_3$, then the subgraph on $N_i$ should be a collection of disconnected cliques. Now because $N_i\cap N_j$ has at most one vertex, every edge that we we encounter in $N(v_i)$ subgraphs is visited only once, so the this step takes $O(n^2)$ time in total.


\paragraph{step 2} For each $v\in N_i$, we check if it has at most one edge to each $N_j$ where $v\notin N_j$ and $j\neq i$. This concluded step $2$.

If some $v\in N_i$ violates this for some $N_j$, $v$ together with $v_j$ and its two neighbors in $N_j$ form a $4$-cycle or diamond. This part takes $O(t+\sum_v deg(v))= O(n^2)$ time as we visit each edge at most twice. 

By the end of these two steps, we know that none of the $v_i$s is in a $4$-cycle or diamond. 
If $t>T$ for some $T$ that we set later, then we remove this independent set and recurs. We can do this at most $n/T$ times, and in that case we spend $O(n^3/T)$ time. 

Suppose that $t<T$. Note that $V(G)\setminus I$ can be written as a union of cliques where every two cliques have at most one node in common. We are going to explain the reasoning behind step $3$ here, before we go into details of the algorithm. We are going to check if any two cliques contain a $4$-cycle or diamond. Consider two cliques $C$ and $C'$. If they share a node, there must be $i\neq j$ where $C\subseteq N_i$ and $C'\subseteq N_j$ and $C\cap C'=\{u\}$, where $u$ is the common neighbor of $v_i$ and $v_j$. If there is an edge between $v\in C$ and $w\in C'$, $v,w\neq u$, then $v_ivwu$ forms a diamond. So unless there is a diamond in the graph, there are no edges between $C\setminus u$ and $C'\setminus u$. Now suppose that $C$ and $C'$ don't have any nodes in common. Then as mentioned before in step 2, each node $v\in C$ has at most one neighbor in $C'$. So if there are two edges between $C$ and $C'$, their endpoints are different, and they form a $4$-cycle. So if there are no $4$-cycles or diamonds in the graph, there is at most one edge between $C$ and $C'$. This follows up to step 3 below. 

\paragraph{step 3} We check if there is more than one edge between any two cliques. We can do this by having a table $T$ with rows and columns indexed by the cliques, and we scan edges one by one and mark the entry corresponding to cliques $C$ and $C'$ if this edge is between $C$ and $C'$. If we find two edges between two cliques, we have a diamond or a $4$-cycle. 
This step takes $O(n^2)$ time, as we visit each edge at most once. In addition, we know the edge between any two cliques (if it exists). Note that by the end of this step we know that there is no diamond or $4$-cycle in the union of exactly $2$ cliques. This concludes step $3$.

As mentioned before, for any two non-intersecting cliques we can have at most one edge between them. Call these edges non-clique edges. The rest of the edges are in cliques. Also since each node has at most $1$ common neighbor with each $v_i$, each node has at most $t-1$ non-clique edges attached to it. 

\paragraph{step 4} First for each node $u$ that is a common neighbor of $v_i$ and $v_j$ for some $i,j$ do the following: Let $C\in N_i$ and $C'\in N_j$ be the cliques that have $u$. For each $C''\notin \{C,C'\}$, check in constant time if there is $w\in C''$ which has an edge to $C\setminus u$ and $C'\setminus u$, using the table $T$. If such $w$ exists, then we have a $4$-cycle  using those edges and $u$. This takes linear time for each $u$, and hence $O(n^2)$ in total. Note that by the end of this step, we know that there is no $4$-cycle or diamond is any $3$ cliques with at least two of them having an intersection. This is because if there is a pattern in $C\cup C'\cup C''$ where $C\cap C'=u$, then there must be a node $v$ in $C''$ that has a neighbor in $C$ and a neighbor in $C'$, and so $u,v$ and the two neighbors form a $4$-cycle or diamond, and that's what we detect in this step. 

\paragraph{step 5} 
We are going to detect $4$-cycle or diamonds that are in exactly $3$ non intersecting cliques.
Note that diamond has a (non-induced) $4$-cycle as its subgraph, and since there is at most one edge between any two cliques, we must have exactly one clique edge in the not-necessarily-induced $4$-cycle contained in the pattern ($4$-cycle or diamond). Now it is easy to see that we can't have a diamond in the union of three non-intersecting cliques if each two has at most one edge between them. For each non-clique edge $e=uv$, do the following: In $O(t)$ time, find all cliques that $u$ and $v$ both have neighbors in it and check if their neighbors are different in that clique. If they are, we have a $4$-cycle. This takes $O(nt^2)$ time. 

\paragraph{step 6}Now the only possibility for a $4$-cycle or diamond is that each of its vertices are in a different clique. So we can delete all the clique edges, and we end-up with a graph with nodes of degree at most $t$. Moreover, since there is at most one edge between any two cliques, in this graph we have at most $t^2 $ edges. First, for every pair of nodes $u,w$ such that $u$ and $w$ belong to different cliques, we define $L(u,w)$ be the list of their common neighbors through non-clique edges. We can compute all the $L(u,w)$s as follows: for every node $z$ and for every two neighbors $u,w$ of $z$ through non-clique edges, put $z$ in $L(u,w)$. This takes $O(nt^2)$ time. Now for every non-clique edge $e=uw$, and for every node $z$, see if $z$ is adjacent to exactly one of $u$ and $w$ through non-clique edges. Suppose it is attached to $u$. Then see if $|L(z,w)|>1$. If so, take a node $u\neq z'\in L(z,w)$, and $zwuz'$ form a $C_4$ or diamond since $zw$ is not an edge. This takes $O(nt^2)$ time. 
Since in this case $t\le T$, we spend $O(n^3/T+nT^2)$ time in total, and if we set $T=n^{2/3}$, we get $O(n^{7/3})$ running time.





\subsubsection{Proof of Theorem \ref{thm:clique}: Detecting $\{C_4,K_4\}$ in $O(n^{7/3})$ time.}
Suppose we want to find an induced $C_4$ or $K_4$ in the host graph $G$.
Let $I=\{v_1,\ldots,v_t\}$ be a maximal independent set. Note that we can find such $I$ in $O(n^2)$ time. For each $i$, Let $N_i=N(v_i)$. Note that $I\cup (\cup_{i=1}^t N_i)=V(G)$ because of the maximality of $I$. For each $u\notin I$, recall that $N_I(u)=\{v_i| u\in N_i\}$ is the set of nodes in $I$ adjacent to $u$.

For each vertex $u$ we define a set $N_{ext}(u)$ as follows: $v\in N_{ext}(u)$ if $uv$ is an edge and there exists $i$, such that $u\in N_i$ and $v\notin N_i$. We call the edge $uv$ an \textit{external} edge (since $v$ is an \textit{external} neighbor of $u$ with respect to the set $N_i$). If $v\in N_{ext}(u)$ and $u\in N_{ext}(v)$, we call the edge $uv$ a \textit{fully external} edge. We have the following simple observation from these definitions.


\begin{observation}
\label{obs:not-f-ext}
If $uv$ is not a fully external edge, then there is $i$ such that $u,v\in N_i$.
\end{observation}

Before proceeding to the algorithm, we prove the following lemma which helps us categorize the patterns ($K_4$ or $C_4$) in $G$.

\begin{lemma}
\label{lem:categorize}
If $H$ is an induced $C_4$ or $K_4$ that does not have any vertex in $I$, then it is of one of the following types. 
\begin{enumerate}
    \item There are two indices $i$ and $j$ such that $V(H)\subseteq N_i\cup N_j$.
    \item $H$ is isomorphic to $K_4$, and there is a node $a\in V(H)$, such that $V(H)\setminus \{a\} \subseteq N_{ext}(a)$.
    \item $V(H)=\{a,b,c,d\}$ and $H$ is isomorphic to $C_4$, with $ab,bc,cd,da\in E(H)$ such that $d,b\in N_{ext}(a)$ and $a,c\in N_{ext}(b)$.
    \item $V(H)=\{a,b,c,d\}$ and $H$ is isomorphic to $C_4$, with $ab,bc,cd,da\in E(H)$ such that $a,c\in N_{ext}(b)\cap N_{ext}(d)$.
\end{enumerate}
\end{lemma}
\begin{proof}
Suppose $V(H)=\{a,b,c,d\}$ where $ab,bc,cd,da\in E(H)$. So we have that either $ac,bd\in E(H)$ or $ac,bd\notin E(H)$. If none of $ab,cd$ are fully external, by Observation \ref{obs:not-f-ext} there are $i$ and $j$ such that $a,b\in N_i$, $c,d\in N_j$ and so $H$ is of type $1$. Similarly, if none of $bc, da$ are fully external, $H$ is of type $1$. 

So WLOG we can assume that $ab,ad$ are fully external. First suppose that $H$ is isomorphic to $K_4$. If $c\in N_{ext}(a)$, then $H$ is of type $2$. Otherwise $N_I(a)\subseteq N_I(c)$. Moreover, similar to above, if both of $ac,bd$ are fully external, $H$ is of type $1$. Otherwise, since $ac$ is not fully external, $bd$ is. Now if $c\in N_I(b)$ or $c\in N_I(d)$, $H$ is of type $2$. Otherwise $N_I(b)\subseteq N_I(c)$ and $N_I(d)\subseteq N_I(c)$. If two of the sets $N_I(a), N_I(b),N_I(d)$ intersect, then $H$ is of type $1$. Otherwise, these sets are disjoint, so $N_I(c)$ is not equal to any of them. So $a,b,d\in N_{ext}(c)$, and $H$ is of type $2$.

Now suppose that $H$ is an induced $C_4$, and recall that $ab,ad$ are fully external.
If $dc$ and $cb$ are both fully external, then $H$ is of type $4$ by the definition of fully external edges. If $bc$ is fully external and $dc$ is not fully external, then $H$ is of type $3$, since by fully externality of $ab$ and $ad$ we have that $b,d\in N_{ext}(a)$, and by fully externality of $ab,bc$ we have that $a,c\in N_{ext}(b)$. Similarly if $dc$ is fully external and $bc$ is not fully external, then $H$ is of type $3$.

So assume that none of $bc$ and $dc$ are fully external. 
If $N_I(d)=N_I(c)$, then since $bc$ is not fully external, $N_I(c)\cap N_I(b) \neq \emptyset$, so there is $i$ such that $b,c,d\in N_i$ and so $H$ is of type $1$. So we can assume that $N_I(d)\neq N_I(c)$ and $N_I(c)\neq N_I(b)$. 

Now if $c\in N_{ext}(d)$, $H$ is of type $3$: $c,a\in N_{ext}(d)$ and $b,d\in N_{ext}(a)$. Otherwise we must have $d\in N_{ext}(c)$ since $N_I(c)\neq N_I(d)$. Similarly, if $c \in N_{ext}(b)$, $H$ is of type $3$. So we can assume that $b\in N_{ext}(c)$ since $N_I(c)\neq N_I(b)$. So $H$ is of type $4$.
\end{proof}

Now we describe the algorithm. We first look for a $C_4$ or $K_4$ that has a vertex in $I$, and then give an algorithm for detecting the patterns of Lemma \ref{lem:categorize} which are patterns with no vertex from $I$.

\paragraph{Step $1$.} Let $L(u,w)$ be a table initialized empty, which is going to stay empty for some pairs $u,w$, or contain the common neighbors of $u,w$. At this step, we are going to fill $L(v_i,v_j)$ for all $i,j\in [t]$. To do so, we first make a pruned adjacency list for every node to only contain the vertices in $I$. 
Then for each $x\notin I$, we put $x$ in $L(v_i,v_j)$ if $x$ is a common neighbor of $v_i$ and $v_j$. If by doing that $|L(v_i,v_j)|=2$, we check if $x$ is adjacent to the other member $y$ of $L(v_i,v_j)$, because otherwise $v_ixv_jy$ forms an induced $C_4$. Moreover, if at some point for some $i,j$,  $L(v_i,v_j)\ge 3$, we have an induced $C_4$ or $K_4$: if there is a non-edge in $L(v_i,v_j)$ like $xy$, then $v_ixv_jy$ forms an induced $C_4$. Otherwise there is a triangle in $L(v_i,v_j)$, and so this triangle with $v_i$ forms a $K_4$. So either we find a $K_4$ or $C_4$, or for all $i,j$ we have $|L(v_i,v_j)|\le 2$, and if $|L(v_i,v_j)|= 2$, the two common neighbors of $v_i$ and $v_j$ are adjacent. 
Note that the total running time is $\sum_x N_I(x)^2=\sum_{i,j} |L(v_i,v_j)|\le O(n^2)$. 

\paragraph{Step $2$.} For every $i\in [t]$, we run the algorithm of Lemma \ref{lem:triangle-and-c4} on the subgraph induced by $N_i$ for detecting a triangle or induced $C_4$. Note that if the algorithm returns one of these patterns, we are done. Otherwise, we have that $|E(G[N_i])|\le |N_i|^{1.5}$. This step takes $\sum_i O(|N_i|^2)$ time.

\begin{claim}
\label{claim:sum-of-ni-2}
$\sum_i |N_i|^2\le O(n^2)$.
\end{claim}
\begin{proof}
Let $l(u,w)=|N_I(u)\cap N_I(w)|$ be the number of $i\in [t]$ such that $u,w\in N_i$. Then by a double counting argument we have that $\sum_i |N_i|^2=\sum_{u,w} l(u,w)$. Now by Step $1$, for any pair $i,j\in [t]$, we have $|L(v_i,v_j)|\le 2$, so there is at most one pair $(u,w)$ such that $u,w\in N_i\cap N_j$. So $l(u,w)^2$ pairs in $I\times I$ are assigned to $(u,w)$ this way, and are not assigned to any other pair. So $\sum_{l(u,w)>1} l(u,w)\le \sum_{l(u,w)>1} l(u,w)^2 \le t^2\le n^2$. So $\sum_{u,w} l(u,w) = \sum_{l(u,w)=1} l(u,w)+\sum_{l(u,w)>1} l(u,w) \le 2n^2$. 
\end{proof}
\paragraph{Step $3$.} In this step, we fill $L(v_i,w)$ for $i\in [t]$ and $w\notin N_i$. For each $i\in [t]$, for every $x\in N_i$, we scan all neighbors of $x$, and if $w$ is a neighbor of $x$ outside $N_i$, we add $x$ to $L(v_i,w)$. If at some point $L(v_i,w)$ has exactly two nodes in it, we check if they are connected. If they are not connected, they form an induced $C_4$ with $v_i$ and $w$. If at some point $L(v_i,w)>2$, we stop since we have an induced $C_4$ or $K_4$: if there is a non-edge in $L(v_i,w)$, this non-edge with $v_i$ and $w$ form an induced $C_4$. Otherwise there is a triangle in $L(v_i,w)$, and it forms a $K_4$ with $v_i$.

We claim that the running time of this step is at most $O(tn+\sum_i |N_i|^{1.5})\le O(tn+\sum_i |N_i|^{2})\le  O(n^2)$ where for the second inequality we use Claim \ref{claim:sum-of-ni-2}.
First note that we don't update $L(v_i,w)$ if $w\in N_i\cap N(x)$. So the number of times that we don't update the table in our search is the number of internal edges in $N_i$s for all $i$, which is $O(\sum_i |N_i|^{1.5})$ by step $2$. The rest of the running time is $\sum_{w,i, w\notin N_i} |L(v_i,w)|$, and since $|L(v_i,w)|$ never reaches $4$ (we stop before it happens) this sum is $O(tn)$. 

If at the end of this step we don't find an induced $C_4$ or $K_4$, then for every $w\notin N_i$ we have $|L(v_i,w)|\le 2$. 

Note that after steps $1,2$ and $3$, if we find no induced $C_4$ or $K_4$, it means that there is no induced $C_4$ or $K_4$ that has at least one vertex in $I$. If $|I|\ge T$, where we define $T$ later, then we remove $I$ from the graph and recurse. We can do this at most $n/T$ times. Note that the running time of steps $1,2$ and $3$ is $O(n^2)$, so the running time is at most $O(n^3/T)$ in this case.

Now suppose that $|I|<T$. From step $3$ we know that $|N_{ext}(w)|\le 2t$ for all $w\notin I$. Moreover, in step $3$ we go over all of the neighbors of all nodes outside $I$, so we can assume that we have the set $N_{ext}(w)$ for all $w\notin I$. In the following steps, we are going to detect patterns of Lemma \ref{lem:categorize}.

\paragraph{Step $4-1$: Type $1$.} First we detect $K_4$s of this type. Note that $N_i$ doesn't have a triangle, so if there is a $K_4$ in $N_i\cup N_j$, it must have exactly two nodes in each set. So for each pair $i,j$, for each node $a\in N_i\setminus N_j$, we check if $a$ has two neighbors in $N_j$ by checking if $|L(a,v_j)|= 2$. If so, then let these neighbors be $b$ and $c$. We check if $b,c\notin N_i$, and if so, we check if they have a common neighbor in $N_i$ other than $a$, by checking if $|L(b,v_i)\cap L(c,v_i)|\ge 2$. If they do have a common neighbor $d$, we check if $abcd$ form a $K_4$. 

The process for each $i,j,a$ takes $O(1)$ time, so in total this takes $O(nt^2)$ time. 

Now we look for induced $C_4$s of this type. 
First note that there is no induced $C_4$ with exactly one node in $N_i$. To see this, suppose this is true, and $a$ is the node in the induced $C_4$ in $N_i$ and not in $N_j$, and $b,c$ are its neighbors in $N_j$. Then since $a\notin N_j$ and has two neighbors in $N_j$ ($|L(a,v_j)|=2$), they must be adjacent (we checked this in step $3$) which is a contradiction.
So any induced $C_4$ has exactly two nodes in $N_i$ and two nodes in $N_j$. The edges of the $C_4$ are either all between $N_i$ and $N_j$, or exactly two of them are between $N_i$ and $N_j$. The first case cannot happen, because if $a,d\in N_i$ and $b,c\in N_j$, then since $b,c$ are $a$'s neighbor in $N_j$, they must be adjacent. So the only case is that $ab,bc,cd,ad$ are edges. To detect these types of patterns, we do the following.

For every $i\in [t]$, for every edge $e=ad\in N_i$, and for every $i\neq j\in [t]$, we check if $a$ and $d$ make an induced $C_4$ with any two of the vertices in $L(v_j,a)\cup L(v_j,d)$ which is a set of size at most $4$. This takes $\sum_i t|N_i|^{1.5}\le t\sqrt{(\sum_i |N_i|)(\sum_i |N_i|^2)}\le t\sqrt{(n+t^2)n^2}\le tn^{1.5}+t^2n$. Note that we use the inequality $\sum_i |N_i|\le t^2+n$ which is true since $|N_i\cap N_j|\le 2$ for all $i,j$.

\paragraph{Step $4-2$: Type $2$.} For every node $u$, we run the algorithm of Lemma \ref{lem:triangle-and-c4} on the subgraph induced by $N_{ext}(u)$. 

Since $|N_{ext}(u)|\le 2t$, this takes $O(t^2)$ time for every $v$, and so this step takes $O(nt^2)$ time. Note that if we don't find any patterns in this step, there is no pattern of type $2$ in the graph. 

\paragraph{Step $4-3$: Type $3$.}
For each fully external edge $e=ab$, we do the following: for each $i\in [t]$ that $a,b\notin N_i$, we look at the set of neighbors of $a$ and $b$ in $N_i$ (i.e. $L(v_i,a)\cup L(v_i,b)$), and see if they form an induced $C_4$ with $a$ and $b$. Note that these neighbors are at most $4$ in total. So this step takes $O(nt^2)$ time, because the number of fully external edges is $\sum_{u\in V(G)}|N_{ext}(u)|\le O(tn)$.

\paragraph{Step $4-4$: Type $4$.}
We fill a table similar to $L$. We want $L_{ext}(w,u)$ include all nodes $z$, such that $w,u\in N_{ext}(z)$. To do this, for every node $z\in V(G)$, we put $z$ in $L_{ext}(w,u)$, for all $w,u\in N_{ext}(z)$. This takes $O(t^2)$ per $z$, and so $O(nt^2)$ time in total. 

Now for every pair $(w,u)$ such that $wu\notin E(G)$ and $|L_{ext}(w,u)|\le 4$, we check if there is a non-edge or a triangle in $L_{ext}(w,u)$. If there is a non-edge in $L_{ext}(w,u)$, then this non-edge with $(w,u)$ form an induced $C_4$ of type $4$. If there is a triangle in $L_{ext}(w,u)$, then this triangle with $w$ forms a $K_4$.

Afterwards, if there is a pair $(u,w)$ such that $uw\notin E(G)$ and $|L_{ext}(w,u)|\ge 4$, then $L_{ext}(u,w)$ either contains a $K_4$, or a non-edge, and in either case, we find a $K_4$ or $C_4$. This processing takes $O(n^2)$ time.

\paragraph{Running time.} Note that in step $4$, we have $t\le T$. So the running time of the algorithm is $O(n^3/T+nT^2+n^{1.5}T)$. Setting $T=n^{2/3}$, the running time is equal to $O(n^{7/3}).$




\subsubsection{$\{H,\bar{H}\}$ Detection}
\label{sec:complement}
In this section we prove Theorem \ref{thm:complement}.

\begin{restatable}{theorem}{complement}
\label{thm:complement}
There is an $O(n^2)$ time algorithm that solves Paired Pattern Detection for $\{H,\bar{H}\}$ for any $4$-node $H$.
\end{restatable}

\begin{figure}
    \centering
    \includegraphics{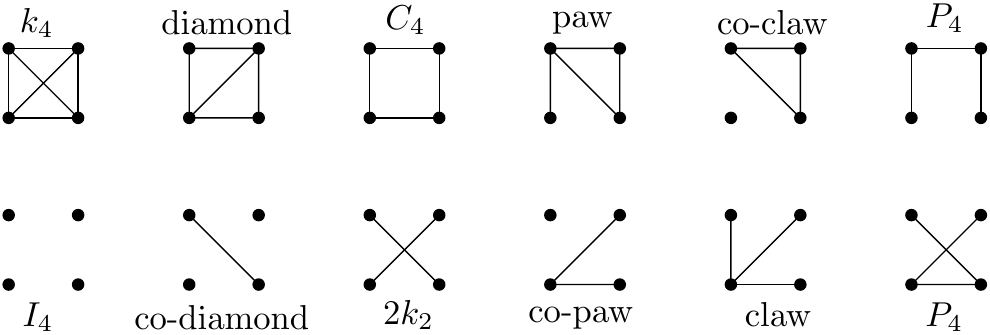}
    \caption{$4$-node patterns paired with their complement.}
    \label{fig:4nodepat}
\end{figure}

There are $6$ cases for the set $\{H,\bar{H}\}$ that are depicted in Figure \ref{fig:4nodepat}. From Ramsey theory results we know that a graph of size $n$ has an independent set or a clique of size $O(\log{n})$. So if a graph is large enough, it has a $4$-clique or a $4$-independent set. For $H=P_4=\bar{H}$, Corneil et al \cite{4path} proved that induced $P_4$ can be detected in linear time. So we are left with $4$ cases, and we prove these $4$ cases separately. 
Even though each case has different techniques, for most of them we start by detecting a clique (or independent set), trying to find an $H$ or $\bar{H}$ that intersects with the clique and then if we find no such pattern, we prove the graph has a specific structure. 
Our algorithms are technically involved similar to that of \cite{4path}. We leave it as an open problem that whether there is a general algorithm that can be used for all $4$-node patterns $H$.

We denote our $n$-node host graph by $G$. Note that finding $\{H,\bar{H}\}$ in $G$ is equivalent to finding $\{H,\bar{H}\}$ in $\bar{G}$. Throughout the rest of this section, we assume that $G$ has at least $31$ nodes. We prove the two cases of $\{diamond, co-diamond\}$ and $\{paw, co-paw\}$ here and prove the other two cases in the appendix. We first state the proof of the following known lemma. 
\begin{lemma}
\label{lem:ramsey}
There is an algorithm that detects a $K_4$ or $I_4$ in an $n\ge 31$ node graph in $O(n)$ time.
\end{lemma}
\begin{proof}
Consider a node $v$ in the graph $G$. If $d_G(v)< (n-1)/2$, we consider $\bar{G}$. So WLOG we assume that $d_G(v)\ge (n-1)/2\ge 15$. So by Theorem \ref{thm:3-node-pairs} in $O(d_G(v))$ time we can find a triangle or $I_3$. If we find a triangle, this triangle with $v$ forms a $K_4$. Otherwise we have an independent set $u,w,z$. Check if a node in $N(v)$ is not attached to any of these three nodes in $O(3d(v))$ time. If such node exists, we have a $I_4$. So assume that for any node $x\in N(v)\setminus\{u,w,z\}$, $x$ is attached to at least one of $u,w,z.$ So WLOG we can assume that $z$ is attached to at least $(d(v)-3)/3\ge 4$ nodes in $N(v)$. If there is an edge among these $4$ nodes, the edge with $v$ and $z$ form a $K_4$. Otherwise, these $4$ nodes form a $I_4$. The runtime is $O(d(v))=O(n)$.
\end{proof}

\paragraph{Diamond and co-diamond}
We want to detect diamond or co-diamond in the graph $G$. We run the algorithm of Lemma \ref{lem:ramsey} on $G$. WLOG, we can assume that the algorithm finds a $K_4$ (If it finds a $I_4$ in $G$, we consider $\bar{G}$ in which we will have a $K_4$). In $O(n^2)$ time, we can turn this $K_4$ into a maximal clique $C$: For every node, we check if it is attached to all the nodes in the clique and if it is, we add it to the clique. 

Now since $C$ is a maximal clique, for every node $u\notin C$ we know that $u$ is not adjacent to at least one node $u'$ in $C$. If $u$ is adjacent to at least two nodes $v,w$ in $C$, then $u,v,u',w$ form a diamond. So we check if $d_C(u)$ is at most $1$ for every $u$, and if for some $u$ it is not, we find a diamond. This takes $O(n^2)$ time. 

Now we show that either we can find a co-diamond or $G\setminus C$ is a clique. Let $u_1,u_2$ be two nodes outside $C$ that are not adjacent. Since they are adjacent to at most one node in $C$ and $|C|\ge 4$, there are two nodes $v,w$ in $C$ such that $u_1,u_2$ are not adjacent to either. So $u_1,u_2,v,w$ form a co-diamond. 

So we can assume that there is no non-edge in $G\setminus C$, hence $G\setminus C$ is a clique. This means that there is no co-diamond in $G$ since the subgraph induced on any $4$ nodes has at least 2 edges. Now if there is a diamond in $G$, it must have 3 nodes in $G\setminus C$ and one node in $C$. This is because if it has two nodes in $C$, one of the two nodes in $G\setminus C$ is adjacent to both of the nodes in $C$, And if it has three nodes in $C$, the node in $G\setminus C$ is adjacent to two of the nodes in $C$. So in any case, there is a node in $G\setminus C$ that has two neighbors in $C$, which contradict our earlier assumption. 

Now to look for diamonds with three nodes in $G\setminus C$, we check the number of neighbors of every node in $C$. If a node in $C$ has at least $2$ neighbors in $G\setminus C$ and it is not adjacent to all the nodes in $G\setminus C$, then we have a diamond. This step takes $O(|C|\cdot|G\setminus C|)=O(n^2)$, and overall the algorithm runs in $O(n^2)$ time.

\paragraph{Paw and co-paw}
First suppose that $G$ is not connected, and suppose $G_1,\ldots,G_t$ are its connected components. Since the paw is a connected pattern, if $G$ has a paw, it is in one of the connected components. To detect any co-paws that might have vertices in different connected components, we look for a $P_3$ in each component using Lemma \ref{lem:p3} in total time of $\sum_i O(|V(G_i)|)^2$. If we find a $P_3$ in $G_i$ for some $i$, then we take a node from $G_j$ for some $j\neq i$, and they form a co-paw. Otherwise, if $G$ has a paw or co-paw, it is completely contained in a connected component. We recurse on each component, and by induction, the total running time will be $\sum_i O(|V(G_i)|)^2=O(n^2)$. 

If $\bar{G}$ is not connected, we can do the same. So we suppose that both $G$ and $\bar{G}$ are connected. 

We run the algorithm of Lemma \ref{lem:ramsey} on $G$. WLOG, we can assume that the algorithm finds a $K_4$. In $O(n^2)$ time, we can turn this $K_4$ into a maximal clique $C=\{v_1,\ldots,v_t\}$ for $t\ge 4$. If there exists $u\notin C$ that has at least one neighbor $v_i\in C$ and that there are at least two nodes $w_1,w_2\in C$ that are not adjacent to $u$, then $u,v_i,w_1,w_2$ form a paw. So we can check the neighbors of each $u\notin C$ in total time of $O(n^2)$, and if we don't find a paw, then any $u\notin C$ is either adjacent to none of the vertices in $C$, or is not adjacent to exactly $1$ node in $C$ (i.e. $d_C(u)$ is either $0$ or $|C|-1$). 

Let $S=\{u\in V(G)\setminus C| d_C(u)=0\}$ and $T=\{u\in V(G)\setminus C| d_C(u)=|C|-1\}$. If $|V(G)|>|C|$, since $G$ is connected we have that $|T|\ge 1$. If $|S|\ge 1$, we show that we find a paw in $S\cup T$: Since $G$ is connected, there is at least one edge between $S$ and $T$. Suppose $s\in S$ and $t\in T$ are adjacent. Let $v_i,v_j\in C$ be two nodes that are adjacent to $t$. Then $s,t,v_i,v_j$ form a paw. 

So assume that $S=\emptyset$. Since vertices in $T$ are adjacent to all but one node in $C$, we can decompose $T$ as follows. Let $T_i$ be the set of nodes in $T$ that are not adjacent to $v_i$. First check in $O(\sum_i |T_i|^2)=O(n^2)$ time if there is an edge in one of $T_i$s. Suppose $T_i$ has an edge $uw$. Then for some $j\neq i$, $u,w,v_i,v_j$ form a paw. 

So we can assume each $T_i$ is an independent set for all $i$. Now check the edges between these independent sets. Suppose that there is an edge between $u\in T_i$ and $w\in T_j$. Let $z\neq i,j$. Then $v_i,v_z,u,w$ form a paw.

So we can assume that for any $i,j$, $T_i\cup T_j$ is a complete bipartite graph. This defines all the edges of the graph. Now note that $G$ is a complete $t$-partite graph: partition $i$ is $T_i\cup v_i$. A complete $t$-partite graph doesn't have induced $P_3$ which is a subgraph of both paw and co-paw, and hence $G$ has no paw or co-paw. The total running time is $O(n^2)$.

\section{Lower bound for induced $C_4$: proof of Theorem \ref{thm:c4hard}}\label{sec:c4lb}
Here we reduce $4$-hyperclique in $3$-uniform hypergraphs to detecting an induced $C_4$ in an $m$-edge graph.

\begin{claim}
Under the $3$-uniform $4$-hyperclique hypothesis, induced $C_4$-detection in $n$-node graphs with $O(n^{3/2})$ edges requires $n^{2-o(1)}$ time on the word-RAM model with $O(\log n)$ bit words.
\end{claim}

\begin{proof}
Without loss of generality, we are given a $4$-partite $3$-uniform hypergraph $G$ with parts $V_0,V_1,V_2,V_3$, and we want to determine if there exist $v_i\in V_i$ for each $i\in \{0,1,2,3\}$ so that every triple $(v_i,v_j,v_k)$ for $i\neq j, i\neq k, j\neq k$ is a hyperedge in $G$.

\begin{figure}[h]\centering
\includegraphics[width=11cm]{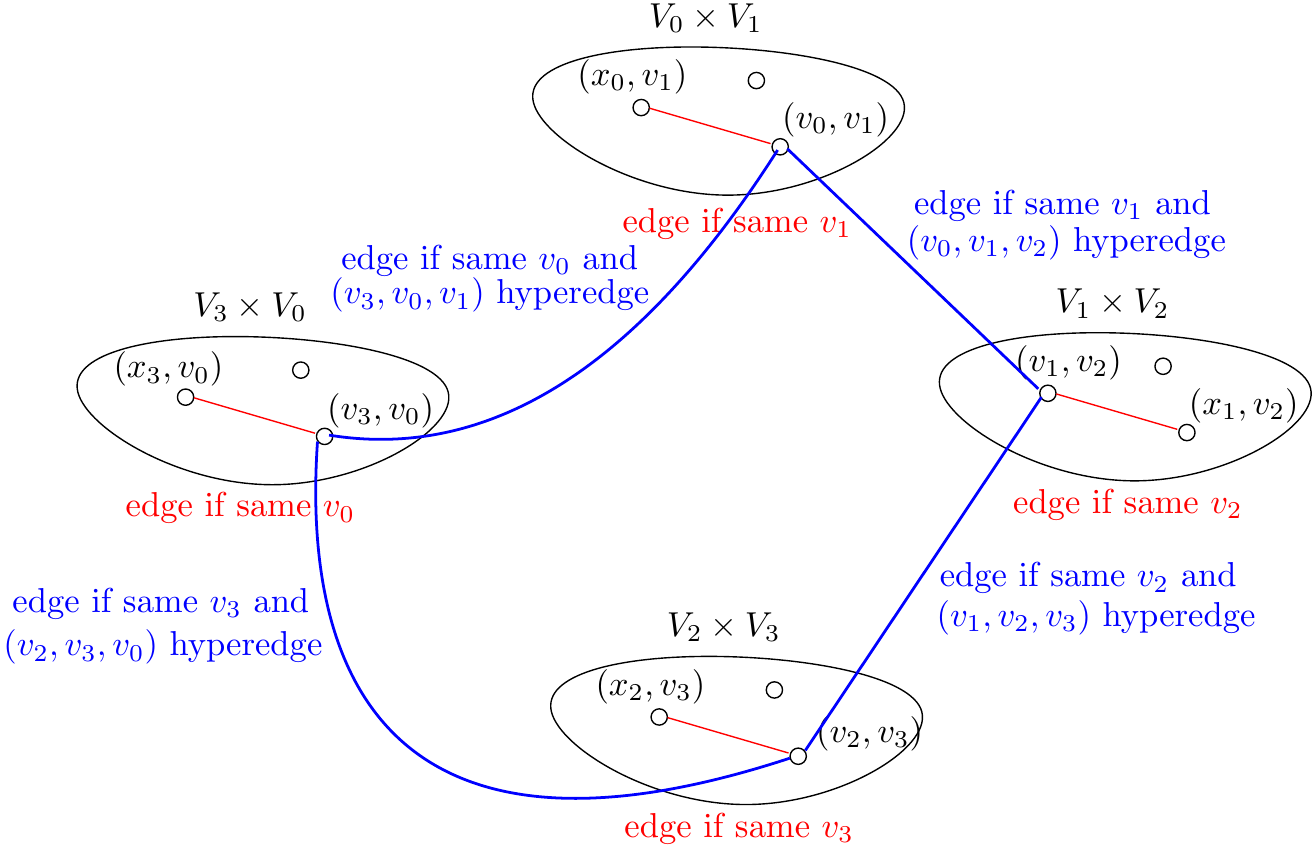}
\caption{A depiction of the lower bound construction for induced $C_4$.}
\label{fig:inducedc4}
\end{figure}

We create an undirected graph $G'$ as follows. See Figure \ref{fig:inducedc4}.
The vertices of $G'$ are all pairs $(x,y)$ such that $x\in V_i$, $y\in V_{i+1}$ for $i\in \{0,1,2,3\}$ where all indices are taken modulo $4$.

The edges are as follows:
\begin{itemize}
\item For every $i\in \{0,1,2,3\}$ and any two nodes $(x,y),(x',y')\in V_i\times V_{i+1}$, there is an edge between them if $y=y'$.
\item For every $i\in \{0,1,2,3\}$ and any two nodes $(x,y)\in V_i\times V_{i+1}$ and $(x',y')\in V_{i+1}\times V_{i+2}$, there is an edge between them if $y=x'$ and $(x,y,y')$ is a hyperedge in $G$.
\end{itemize}

First suppose that $v_0\in V_0, v_1\in V_1,v_2\in V_2, v_3\in V_3$ form a $4$-hyperclique in $G$. Then since $(v_0,v_1,v_2), (v_1,v_2,v_3), (v_2,v_3,v_0), (v_3,v_0,v_1)$ are all hyperedges in $G$, the following is an induced $C_4$ in $G'$: $(v_0,v_1),(v_1,v_2),(v_2,v_3),(v_3,v_0)$.

Now suppose that there is an induced $C_4$, $H$ in $G'$. If for each $i\in \{0,1,2,3\}$, $H$ has a vertex in $ (v_i,v_{i+1})\in V_i\times V_{i+1}$, then for each $i$, $(v_i,v_{i+1},v_{i+2})$ is a hyperedge, so that $H$ corresponds to a $4$-hyperclique.

Suppose now that $H$ doesn't have a node in every single one of the four $V_i\times V_{i+1}$. Then for some $i$, $H$ contains at least two nodes.

If $H$ contains three nodes $(v_i,v_{i+1}),(v'_i,v'_{i+1}),(v''_i,v''_{i+1})$ in $V_i\times V_{i+1}$, then these three nodes have two edges between them (since $H$ is a $C_4$).
W.l.o.g. $(v_i,v_{i+1})$ has edges to $(v'_i,v'_{i+1})$ and $(v''_i,v''_{i+1})$.
However, by the construction of $G'$, this means that $v_{i+1}=v'_{i+1}=v''_{i+1}$, and hence $H$ has a triangle, so it cannot be an induced $C_4$.

Suppose that $H$ contains exactly two nodes $(v_i,v_{i+1}),(v'_i,v'_{i+1})\in V_i\times V_{i+1}$.

If these two nodes are connected by an edge, then $v_{i+1}=v'_{i+1}$. The two nodes $(v_i,v_{i+1}),(v'_i,v_{i+1})$ can only have edges to $V_{i+1}\times V_{i+2}$ and to $V_{i-1}\times V_i$, and those two partitions have no edges between them. Thus the only way $(v_i,v_{i+1})$ and $(v'_i,v_{i+1})$ can be part of a $C_4$ is if they are connected to two nodes in the same part, that are also connected to each other.
If the edges of $H$ out of  $(v_i,v_{i+1})$ and $(v'_i,v_{i+1})$ are to 
$V_{i-1}\times V_i$, then they are to nodes $(v_{i-1},v_i)$ and $(v'_{i-1}, v'_i)$ that are supposed to have an edge between them. However then $v_i=v'_i$, so that $(v_i,v_{i+1})$ and $(v'_i,v_{i+1})$ are the same node, so this cannot happen. On the other hand, if the edges of $H$ out of  $(v_i,v_{i+1})$ and $(v'_i,v_{i+1})$ are to 
$V_{i+1}\times V_{i+2}$, then they are two some nodes $(v_{i+1},v_{i+2})$ and $(v_{i+1},v'_{i+2})$ that supposedly have an edge between them. But then $v_{i+2}=v'_{i+2}$ and these two are the same node, so that can't happen.

Thus it must be that $(v_i,v_{i+1})$ and $(v'_i,v'_{i+1})$ are not connected by an edge. Then they have a common neighbor in $V_{i-1}\times V_i$ or $V_{i+1}\times V_{i+2}$ or both.
If there is a common neighbor $(v_{i+1}'',v_{i+2})$ in $V_{i+1}\times V_{i+2}$, then by construction, it must be that $v_{i+1}''=v'_{i+1}=v_{i+1}$, but then there would be an edge between $(v_i,v_{i+1})$ and $(v'_i,v'_{i+1})$, and that cannot happen.
Thus both common neighbors must be in
$V_{i-1}\times V_i$. But then $v_i=v'_i$, and the common neighbors in
$V_{i-1}\times V_i$ must look like $(v_{i-1},v_i)$ and $(v'_{i-1},v_i)$, and thus must be connected by an edge, and hence $H$ is not an induced $C_4$ but a diamond.

Hence any induced $C_4$ must have a node in each $V_i\times V_{i+1}$ and thus corresponds to a $4$-hyperclique.

If $G$ had $n$ nodes, then $G'$ has $N=O(n^2)$ nodes and $M=O(n^{3})$ edges. Any $O(M^{4/3-\eps})$ time algorithm for $\eps>0$ for induced $C_4$ would imply an $O(n^{4-3\eps})$ time algorithm for $4$-hyperclique and would refute the $3$-uniform $4$-hyperclique hypothesis. The same holds for an $O(N^{2-\eps'})$ time induced $C_4$ algorithm for $\eps'=1.5\eps>0$.
\end{proof}
\section*{Open Problems}
Considering the hardness results based on the $k$-clique hypothesis, there is still a gap between lower and upper bounds for Induced Subgraph Isomorphism, where the best algorithm for detecting a $k$-node pattern $H$ runs in $k$-clique detection time and the best lower bound states that detecting $H$ requires the time needed to detect a $\sqrt{k}/\log{k}$-clique. The result of Manurangsi et al. \cite{ManurangsiRS21} suggests that the true running time should be closer to the current upper bound, and so it would be interesting to obtain a similar lower bound to \cite{ManurangsiRS21}, conditioned on $k$-clique hypothesis. 

For non-induced Subgraph Isomorphism, we show that cores as well as complements of paths and cycles are at least as hard to detect as a $t$-clique where $t$ is roughly the size of the maximum clique minor of the pattern. It is an interesting open problem to extend this result to more pattern classes. 

Finally, the true time complexity of induced $4$-cycle detection remains open. While resolving the gap for (induced) Subgraph Isomorphism might seem hard, Paired Pattern Detection can be a good guide on what properties future reductions must posses. Having a general framework for obtaining algorithms for Paired Pattern Detection would be a good step forward for this goal.

\bibliographystyle{plain}
\bibliography{ref}
\section*{Appendix}
\subsection{Lower Bounds}
\begin{proofof}{Lemma \ref{lem:maxcliqueminor}}
First we show that the maximum clique minor of any $k$-node pattern $H$ is a at most $\floor{\frac{k+\omega(H)}{2}}$, where $\omega(H)$ is the size of the maximum clique subgraph of $H$. Consider any clique minor of $H$, it partitions the pattern into connected groups, such that there is at least one edge between each two groups. Let the number of groups with exactly one node be $\alpha$. Note that $\alpha\le w(H)$, since these single nodes must form a clique. Moreover, the $k-\alpha$ nodes remaining must be in groups of size bigger than one, so they create at most $\floor{\frac{k-\alpha}{2}}$ groups. So in total, we have $\floor{\frac{k+\alpha}{2}}\le\floor{\frac{k+\omega(H)}{2}}$ groups. Note that the values of Table \ref{table:maxcliqueminor} match $\floor{\frac{k+\omega(H)}{2}}$.

To give clique minors of size $\floor{\frac{k+\omega(H)}{2}}$ for complements of paths and cycles, we do the following. If $\bar{H}$ is a path, let $\bar{H}=v_1v_2\ldots v_k$ and if it is a cycle let $\bar{H}=v_1v_2\ldots v_kv_1$. First we consider the maximum clique of $H$: if $k$ is even, it is $\{v_1,v_3,\ldots,v_{k-1}\}$, if $\bar{H}=C_k$ for odd $k$ it is $\{v_1,v_3,\ldots,v_{k-2}\}$ and if $\bar{H}=P_k$ for odd $k$ it is $\{v_1,v_3,\ldots,v_{k-2},v_k\}$. 
We have to pair the remaining $k-\omega(H)$ (with possibly one group containing three nodes if $k-\omega(H)$ is odd), such that each group is connected and each two groups have an edge between them. 
To do so, we need to group them in a way that if the size of a group is $2$, then the two nodes are not of the form $v_i,v_{i+2}$ for some $i$ mod $k$, because then none of them are adjacent to $v_{i+1}$ which is in the maximum clique and is the only member of its group. Pairing vertices with this restriction can be easily done except for when the pattern is $\bar{P_4}, \bar{C_4}$ and $\bar{P_5}$. 
\end{proofof}

\begin{proofof}{Lemma \ref{lem:odd}}
Let $\bar{C_k}$ be the complement of the cycle $C_k=v_1v_2\ldots v_kv_1$ for odd integer $k>3$. First we show that the chromatic number of $\bar{C_k}$ is $\frac{k+1}{2}$. This is because $C_k$ doesn't have a triangle, so at most $2$ nodes can be of the same color in $\bar{C_k}$. This means that the chromatic number of $\bar{C_k}$ is at least $\frac{k+1}{2}$, and the following proper coloring has exactly $\frac{k+1}{2}$ colors: Color $v_1$ with $1$, and for $i=1,\ldots \frac{k-1}{2}$, $v_{2i}$ and $v_{2i+1}$ with color $i$. 

To show that $\bar{C_k}$ is color critical, it suffices to prove that the chromatic number of $\bar{C_k}\setminus v_1$ is strictly less than $\frac{k+1}{2}$, which can be easily seen by the coloring given above. 
\end{proofof}

\subsection{Algorithms}

\begin{proofof}{Theorem \ref{thm:3-node-pairs}}
If $n\le 7$, we look at any subgraph of size $3$ in the graph, and check if it is isomorphic to $H_1$ or $H_2$. So suppose $n>7$.

First suppose that $H_2=H_1\setminus e$ for some edge $e=uw$ in $H_1$. Let $v\neq u,w$ be the node in $H_1$ that is not the endpoint of $e$. Let $v',u'$ and $w'$ be the corresponding nodes in $H_2$, so the edges $u'v'$ and $v'w'$ exist if and only if $uv$ and $vw$ exist. There are three cases for the edge pairs $vu$ and $vw$: (1) $vu,vw\in E(H_1)$, (2) $\{vu,vw\}\cap E(H_1)=1$, (3) $vu,vw\notin E(H_1)$. For case (1), look for a node $z$ in the host graph that has degree at least $2$. If such $z$ exists, pick two of its neighbors and the two neighbors and $z$ form $H_1$ or $H_2$. If no such $z$ exists, then the graph doesn't have either of $H_1$ and $H_2$. The other cases are similar. For case (2), the host graph has $H_1$ or $H_2$ if and only if it has a node $z$ with degree at least $1$ and at most $n-2$. For case (3) the host graph has $H_1$ or $H_2$ if and only if it has a node $z$ with degree at most $n-3$. The node $z$ in all the cases can be found in linear time by checking the degree of each node.

The cases that are not covered by the argument above are $\{K_3, I_3\}$, $\{K_3,P_2\cup I_1\}$ and $\{I_3,P_3\}$. The last two cases are complements of each other, so we prove one of them. 

First we show that we can fine induced $\{K_3, I_3\}$ in linear time. WLOG suppose that there is a node $v$ in the host graph with $d(v)\ge 3$ (if no such vertex exists, consider the complement of the host graph). Consider three of the neighbors of $v$, if two of them are attached, we have a triangle. If there is no edge among them, we have a $I_3$. 

Finally, we prove that we can detect induced $\{K_3,P_2\cup I_1\}$ in $O(n^2)$ time. First run the induced $\{K_3, I_3\}$ detection algorithm in linear time. If it finds no pattern, then run the induced $\{I_3,P_2\cup I_1\}$ detection algorithm and we show that we are done: If it also finds no pattern, then there is no $K_3$ or $P_2\cup I_1$ in the graph. If it finds a pattern, since the host graph has no $I_3$, it must be $P_2\cup I_1$. So we are done. 

So suppose that the $\{K_3, I_3\}$ detection algorithm finds a pattern. If it is $K_3$, then we are done. So suppose it is $I_3$. In $O(nt)$ time we can make this independent set into a maximal independent set $S$, where $t$ is the size of $S$. For each $v\notin S$ in the host graph, $v$ is attached to at least one node in $S$. If it is not attached to all the nodes in $S$, then $v$ with a neighbor and a non-neighbor in $S$ form a $P_2\cup I_1$. Let $S'$ be the set of all nodes in the host graph that are not in $S$. We can check the $S$ neighbors of all $v\in S'$ in $O(nt)$ time, and if we don't find any $P_2\cup I_1$, there is no $P_2\cup I_1$ or $K_3$ that contains a vertex in $S$. So we can recurse on $S'$. Note that we spent $O(nt)$ time, and if the algorithm on $S'$ takes $O((n-t)^2)$, then the algorithm takes $O(n^2)$ in total.

\end{proofof}

\begin{proofof}{Theorem \ref{thm:c4-general}}
We are going to prove the theorem for $H$ being the paw and co-claw.
\paragraph{The paw} We first prove two claims.
\begin{claim}
\label{claim:cliquepaw}
If $N(v)$ is a clique for some $v$, then we can determine whether the graph has a $4$-cycle or paw in $O(n^2)$ time. 
\end{claim}
\begin{proof}
First note that we can assume that the graph is connected since both patterns are connected. For each $u\notin N[v]$, if the graph has no paws, we have that either $u$ is attached to all vertices in $N(v)$, or it is attached to none. This is because if it is attached to $w\in N(v)$ and not attached to $z\in N(v)$, then $\{v,u,w,z\}$ induces a paw. So let $S$ be the set of vertices attached to all $N(v)$, and let $T$ be the rest. There is no edge between $T$ and $S$: if $w\in S$ is attached to $u\in T$, then for some $z\in N(v)$, $\{u,v,w,z\}$ induces a paw. Since the graph is connected, this means that $T=\emptyset$. Now if there is an edge $zw$ in $S$, then for some $u\in N(u)$, $\{u,v,w,z\}$ induces a paw. So there is no edge in $S$. So the graph consists of a clique, $N(v)$, and an independent set, $G\setminus N(v)$, with all the edges between them present, and it doesn't have a paw or a $4$-clique. We visited each edge at most once, so we spent $O(n^2)$ time. 
\end{proof}

\begin{claim}
\label{claim:clquepaw2}
If $N(v)$ has an edge, then we can determine whether the graph has a $4$-cycle or paw in $O(n^2)$ time.
\end{claim}
\begin{proof}
The proof is similar to Claim \ref{claim:cliquepaw}. 
Check if $N(v)$ has an induced $\bar{p_3}$ by Lemma \ref{lem:p3}. If it does, then there is a paw in the graph, otherwise, $N(v)$ is a complete $c$-partite graph for some $c$. Now if two parts of this $c$-partite graph has size at least $2$, we have a $4$-cycle. So at most one of them has size at least $2$, and the rest have size $1$. If all parts have size $1$, then $N(v)$ is a clique and by Claim \ref{claim:cliquepaw} we are done. So let the part with size bigger than $1$ be $I$. Let $J=N(v)\setminus I$. Note that since $N(v)$ has an edge, $J\neq \emptyset$. Let $w\in I$. 

For each $u\notin N[v]$, if there is an edge and a non edge in $J \cup w \times u$, then we have a paw since $J \cup w $ is a clique. So $u$ is attached to all vertices in $J \cup w$ or it is adjacent to none. Similar argument works for $J \cup z$ for any $z\in I$. So either $u$ is in a paw and we detect it, $u$ is attached to all $N(v)$, or $u$ is attached to none in $N(v)$. Now if $u$ is attached to $z,w\in I$, then $u,v,z,w$ induce a $4$-cycle. So $u$ is attached to none of the vertices in $N[u]$. We check this for every $u$. If we don't find any pattern, then $N[v]$ is disconnected from the rest of the graph, and hence $G=N[v]$ and we output that we found no patterns.
\end{proof}

We run the algorithm of \cite{non-induced-4cycle} on the host graph to find a non-induced $4$-cycle. If it outputs a pattern, then we either have an induced $4$-cycle, or a triangle, so we can use Claim \ref{claim:clquepaw2}. If it outputs that there is no $4$-cycle, then it means that for every pair of nodes $u,v$, there is at most one node attached to both $u$ and $v$, so $\sum_w N(w)^2=O(n^2)$. So for each node $w$ in the host graph, we check if there is an edge in $N(w)$, and if there is such $w$, we apply Claim \ref{claim:clquepaw2}. Note that in finding such $w$ we spend $\sum_w N(w)^2=O(n^2)$ time. 

\paragraph{Co-claw}
First, we remove any vertices with degree $n-1$ as they cannot be in any of the two patterns. So we can assume that the degree of every vertex is less than $n-1$. Now run the algorithm of Lemma \ref{lem:triangle-and-c4} to detect a triangle or a $4$-cycle. If it outputs no pattern, then there is no $4$-cycle or co-claw in the graph. If it outputs a $4$-cycle, we are done as well. So suppose that it outputs a triangle. 

Next, we expand the triangle into a maximal clique, by visiting each vertex and seeing if it is adjacent to all the vertices in the clique. Let this maximal clique be $C$. Suppose that there is a vertex $v\in C$ with no neighbor in $G\setminus C$. Run the algorithm of Lemma \ref{lem:triangle-and-c4} on $G\setminus C$ to find a triangle or a $4$-clique. If the algorithm outputs a triangle, then this triangle with $v$ forms a co-claw. If the algorithm outputs a $4$-cycle we are done. If it outputs no pattern, then there is no $4$-cycle or co-claw containing $v$, so we can remove $v$ from the graph. So we can assume that for every $v\in C$, $v$ has at least one neighbor in $G\setminus C$. We have spent $O(n^2)$ time so far. 

Now we show that if there is an edge in $G\setminus C$, then we have a $4$-cycle or a co-claw. First we check if there is a node $u\in G\setminus C$ that is not attached to at least three nodes in $C$. If such $u$ exists, then we have a co-claw. So suppose that for each $u\in G\setminus C$, there are at most two nodes in $C$ which are not attached to $u$. Let the set of nodes in $C$ not attached to $u$ be $S(u)$. We compute $S(u)$ for every $u$, in total time of $O(n^2)$. Let $e=uw$ be an edge in $G\setminus C$. If $S(u)\not\subseteq S(w)$ and $S(w)\not \subseteq S(v)$, then there is $v_1,v_2\in C$ where $v_1\in S(u)\setminus S(w)$ and $v_2\in S(w)\setminus S(u)$, and so $v_1,v_2,u,w$ induce a $4$-cycle. So WLOG assume that $S(u)\subseteq S(w)$. There are three cases:

\textbf{(case 1)} $S(u)=S(w)=\{v\}$: Since $v$ has a neighbor in $G\setminus C$, there is $z\in C$, $z\neq u,w$, where $zv\in E$. Since $C$ is a maximal clique, there is $v'\in C$, $v'\neq v$, where $v'z\notin E$. Note that $uv',wv'\in E$. If $zu\in E$, then $v,v',z,u$ is a $4$-cycle. If $zw\in E$, then $v,v',z,w$ is a $4$-cycle. If neither happens, then $v',u,w,z$ induce a co-claw. 

\textbf{(case 2)} $S(u)=S(w)=\{v_1,v_2\}$: Let $v_3\neq v_1,v_2$ be some vertex in $C$. Since $v_3$ has degree at most $n-2$, there is $z\in G\setminus C$ where $zv_3\notin E$. We know that $uv_3,wv_3\in E$, and either $zv_1\in E$, or $zv_2\in E$, since $|S(z)|\le 2$. Wlog assume that $zv_1\in E$. Then either $v_1v_3zu$ is a $4$-cycle, $v_1v_3wz$ is a $4$-cycle, or $uwv_3z$ is a co-claw. 

\textbf{(case 3)} $S(u)=\{v_1\}$, $S(w)=\{v_1,v_2\}$. Let $z$ be a neighbor of $v_1$ in $G\setminus C$. Suppose $z$ is not adjacent to some $v_3\in C$, $v_3\neq v_2$. Then since $uv_3,wv_3\in E$ either $wzv_3v_1$ is a $4$-clique, $zuv_1v_3$ is a $4$-cycle, or $wuv_3z$ is a co-claw. We can check if such $z$ exists in $O(d(v_1)n)$ time. So assume that there is no such $z$. This means that for each $z\in G\setminus C$, either $v_1\in S(z)$, or $S(z)=\{v_2\}$. Since $d(v_1)<n-1$, there exists some $z_1\in N(v_1)$, and so we have that $S(z_1)=\{v_2\}$.
Let $v_3\in C$, $v_3\neq v_1,v_2$. Since $d(v_3)\neq n-1$, there is $z_2\in G\setminus C$, where $v_3\in S(z_2)$. So $S(z_2)=\{v_1,v_3\}$, and hence $z_2v_2\in E$. Now if $z_1z_2\in E$, then $z_1z_2v_1v_2$ is a $4$-cycle. Otherwise, $z_1z_2v_1v_3$ is a co-claw.

So if we find no pattern, we can assume that there is no edge in $G\setminus C$. Since $|S(v)|\le 2$ for all $v\in G\setminus C$, this means that there is no co-claw or $4$-cycle in $G$.
\end{proofof}

\subsection{proof of Theorem \ref{thm:complement}}
\paragraph{Claw and co-claw} 
First we prove the following claim.
\begin{claim}
\label{claim:degree3}
Let $n\ge 11$. If there is a vertex $v\in V(G)$ such that $3\le d_G(v)\le n-4$, then $G$ has a claw or co-claw and we can find it in $O(n^2)$ time.
\end{claim}
\begin{proof}
WLOG we can assume that $d_G(v)\ge \frac{n-1}{2}$ (otherwise consider $\bar{G}$). Let $N_v$ be the set of neighbors of $v$, and let $M_v$ be the rest of the nodes. We know that $|N_v|\ge 5$ and $|M_v|\ge 3$. Suppose that there is a non-edge $uw$ in $M_v$. If there is a node $z\in N_v$ that is attached to both $u$ and $w$, then $z,u,w,v$ form a claw. So we check all node in $N_v$ in $O(n)$ total time, and if we don't find a claw, for all $z\in N_v$, $z$ is not attached to at least one of $u$ and $w$. Since $|N_v|\ge5$, there is one of $u$ and $w$ that is not attached to at least $3$ nodes in $N_v$. Suppose $u$ is not attached to $z_1,z_2,z_3\in N_v$. If there is an edge between $z_1,z_2,z_3$, that edge with $v$ and $u$ form a co-claw. Otherwise, there is no edge between $z_1,z_2,z_3$ and so $z_1,z_2,z_3,v$ form a claw. We can check neighbors of $u$ and $w$ in $O(n)$ time and find a claw.

So if there is no non-edge in $M_v$, there is a triangle in $M_v$ and this triangle with $v$ forms a co-claw. 
\end{proof}
If $n<11$, we check all subgraphs of size $4$ for a claw or co-claw. Otherwise, we check the degree of all nodes and if there is a node $v$ such that $3\le d_G(v)\le n-4$, by Claim \ref{claim:degree3} we can find a claw or co-claw in $O(n^2)$ time. 

So suppose that there is no such vertex $v$. So for each $v\in V(G)$, $v$ is either low degree and $d_G(v)\le 2$ or $v$ is high degree and $d_G(v)\ge n-3$. Suppose that we have both kind of nodes in the graph. WLOG suppose that we have more low degree nodes. Let $v$ be a high degree node. Let $N_v$ be the set of neighbors of $v$. Since most nodes are low degree, there is a low degree node $u$ in $ N_v$ and we can find it in $O(|N_v|)$ time. So there are $z_1,z_2,z_3\in N_v$ that are not attached to $u$ and we can find them in $O(n)$ time. If there is a non-edge among $z_1,z_2,z_3$, then this non-edge with $u$ and $v$ form a claw. Otherwise $z_1,z_2,z_3$ form a triangle, and so they form a co-claw with $u$. 

So suppose that all nodes are low degree or all nodes are high degree. WLOG suppose the former happens. Since the degree of all nodes is less than $3$, there is no claw in the graph. To look for a co-claw, for every node $v$ we check if it is in a triangle in $O(1)$ time. If it is, then since all nodes are low degree there is a node in the graph that has no neighbors in the triangle, and hence we have a co-claw. This take $O(n)$ time, and overall the algorithm takes $O(n^2).$
\paragraph{$C_4$ and $2k_2$}
First we have the following observation that is easily verifiable.
\begin{observation}
\label{obs:clique-or-ind}
\cite{splitgraphs} If $G$ is decomposed into a clique and an independent set (i.e $G$ is a split graph), it doesn't have a $C_4$ or $2k_2$.
\end{observation}

Now we explain the algorithm. If the size of the graph is at most $31$, we check all subgraphs of size $4$ to find a $C_4$ or $2k_2$. So assume that $n>31$.
\paragraph{Step 1.} Run the algorithm of Lemma \ref{lem:ramsey} on $G$ to find a $K_4$ or a $I_4$. WLOG suppose we find a $K_4$. In $O(n^2)$, we can turn this $4$-clique into a maximal clique $C$. Let $T_C$ be the set of nodes that are attached to all but exactly one node of the clique: $d_C(v)=|C|-1$ for $v\in T_C$. Let $S_T$ be the rest of the vertices. So for each $v\in S$, we have $d_C(v)\le |C|-2$. 

\paragraph{Step 2.} Scan $S_C$ for edges in $O(|S_C|^2)$ time. Suppose there is an edge $uv\in S_C$. If there are two nodes $c_1,c_2\in C$ such that $u$ and $v$ are not attached to $c_1$ and $c_2$, then $u,v,c_1,c_2$ form a $2k_2$. If such $c_1,c_2$ don't exist, then there exists $c_1',c_2'\in C$ such that $uc_1',vc_2'\in E(G)$ and $uc_2',vc_1'\notin E(G)$, and so $u,v,c_1,'c_2'$ form a $C_4$. So we scan $S_T$ for edges and if we find an edge, in $O(n)$ time we can find either a $C_4$ or a $2k_2$.

So we can assume that $S_T$ is an independent set. 

\paragraph{Step 3.} Compute $E^*=E_{T_C}$ which is the set of edges with both endpoints in $T_C$. If $E^*=\emptyset$, then we can apply Lemma \ref{lem:c4mainlemma} below. Otherwise, suppose there is an edge $e=uv\in E(T^*)$ such that $N_C(v)\neq N_C(u)$. Then there are nodes $c_1,c_2\in C$, such that $uc_1,vc_2\in E(G)$ and $uc_2,vc_1\notin E(G)$. So $u,v,c_1,c_2$ form a $4$-cycle. 

Now suppose that there is an edge $e=uv\in E(T^*)$ such that $N_C(v)= N_C(u)$. Let $z\in C$ be the vertex $u$ and $v$ are not attached to. Suppose that there is $w\in T_C$ that is attached to $z$, and let $z'\in N_C(w)$. So $uz',vz'\in E(G)$. If $wu\in E(G)$, then $z',z,w,u$ form a $4$-cycle. Similarly if $wv\in E(G)$, then $z',z',w,v$ form a $4$-cycle. Otherwise, $z,w,u,v$ form a $2k_2$. So if we haven't found a $4$-cycle of $2k_2$ so far, we know that all nodes in $T_C$ have the same neighbors in $C$, i.e. they are not attached to the same node $z\in C$.

\paragraph{Step 4.} Let $C'$ be the clique formed as follows: Remove $z$ from $C$, add $u$ and $v$ and then make the clique  maximal by looking for nodes to add from $T_C$. This can be done in $O(n\cdot (|C'|-|C|))$ time. We can also compute $T_{C'}$ from $T_{C}\cup S_C$ in $O(n.(|C'|-|C|))$ time. Note that $z\in S_{C'}$, as it is not attached to $u,v\in C'$. If $T_{C}\cap T_{C'}=\emptyset$, then $T_{C'}\subseteq S_C$, so $T_{C'}$ is an independent set and we can use Lemma \ref{lem:c4mainlemma}. 

Suppose $x\in T_{C}\cap T_{C'}$. Suppose $T_{C'}\cap S_C\neq\emptyset$, and $y\in T_{C'}\cap S_C$. Then $y$ and $x$ have different neighborhood sets in $C'$. So we can repeat steps $2$ and $3$ and find a pattern.

Now suppose that $T_{C'}\cap S_C=\emptyset$. Then $T_{C'}\subseteq T_{C}$. We repeat the steps $2,3$ and $4$ again for $C'$. Note that step $2$ can now be done in $O(n\cdot(|S_{C'}|-|S_{C}|))$, since $S_{C}\subseteq S_{C'}$. Step $3$ can be done in $O(|E_{T_{C}}|-|E_{T_{C'}}|)$ and step $4$ can be done in $O(n+n(|C'|-|C|))$. Moreover, the size of the clique is growing, so we repeat at most $O(n)$ times, and thus we spend $O(|E_{T_{C}}|+n^2)$, until we find a maximal clique $C^*$ that satisfies Lemma \ref{lem:c4mainlemma}.

\begin{lemma}
\label{lem:c4mainlemma}
Let $G$ be a $n$-node graph. Let $C$ be a maximal clique in $G$, let $T$ be the set of nodes $v$ where $d_C(v)=|C|-1$ and let $S$ be the rest of the nodes, i.e for every node $v\in S$ we have $d_C(v)\le |C|-2$. Then if $S$ and $T$ are independent sets, in $O(n^2)$ we can either find a $C_4$ or a $2k_2$ or show that $G$ doesn't have either of these patterns.
\end{lemma}
 
\begin{proof}
Suppose there is $v\in S$ such that $d_T(v)\ge 3$, i.e. $v$ has at least three neighbors in $T$. Let $c_1,c_2\in C$ be two nodes in $C$ that are not attached to $v$. Since each node in $T$ is attached to all but exactly one node in $C$, there is $i\in \{1,2\}$ such that $c_i$ is attached to at least two nodes in $N_T(v)$, say $w_1,w_2$. So $v,c_i,w_1,w_2$ form a $C_4$, and we can find it in $O(n)$ time once we have $v$.

So we can assume that $d_T(v)\le 2$ for all $v\in S$. Suppose there is a $v$ with $d_T(v)=2$ such that $v$ is not attached to at least three nodes $c_1,c_2,c_3$ in $C$. Then again since every node in $T$ is attached to all but exactly one node in $C$, there is $i\in \{1,2,3\}$ such that $c_i$ is attached to both nodes in $N_T(v)$. So $N_T(v), c_i,v$ form a $C_4$, and we can find it in $O(n)$ time once we have $v$.

So we can assume that for all $v\in S$ with $d_T(v)=2$, $v$ is attached to all but exactly two nodes in $C$. First suppose that $|T|\ge 3$. If there is a node $v\in S$ with $d_T(v)=2$, let $u\in T\setminus N_T(v)$. Suppose $v$ is not attached to $c_1,c_2\in C$. If both nodes in $N_T(v)$ are not attached to $c_1$, then both are attached to $c_2$ and $v,N_v(T),c_2$ form a $C_4$. Similarly if both nodes in $N_T(v)$ are not attached to $c_2$, we have a $C_4$. So suppose one is not attached to $c_1$ and the other is not attached to $c_2$. WLOG suppose $u$ is attached to $c_1$. There is $w\in N_T(v)$ that is not attached to $c_1$. So $w,u,v,c_1$ form a $2k_2$. So if $|T|\ge 3$ and there is $v\in S$ with $d_T(v)=2$, we find a $2k_2$ or $C_4$ in $O(n)$ time.

So we can suppose that either there is a node $v\in S$ with $d_T(v)=2$ and $|T|\le 2$, or for all $v\in S$, $d_T(v)\le 1$.

First suppose that $d_T(v)\le 1$ for all $v\in S$. If $d_T(v)=0$ for all $v\in S$, then $S\cup T$ is an independent set and by Observation \ref{obs:clique-or-ind} there is no $C_4$ or $2k_2$ in $G$. So suppose that there is $v\in S$ with $d_T(v)=1$. Let $w\in T$ be the neighbor of $v$. Suppose that $w$ is not attached to $c\in C$. If $vc\in E(G)$, let $c'$ be a node in $C$ that is not attached to $v$. Then $w$ must be attached to $c'$, since $w$ is attached to all but one node in $C$. So $w,v,c,c'$ form a $C_4$ and we can find it in $O(n)$ time. So suppose that $vc\notin E(G)$. If there is $w'\in T$ such that $w'$ is attached to $c$, then since $v$ is not attached to $w'$, $c,w',w,v$ form a $2k_2$ and we can find it in $O(n)$. 

So suppose that for all nodes $t\in T$, $t$ is not attached to $c$, and is attached to all nodes $C\setminus c$. Now if there is $v'\in S$, such that $d_T(v')=1$ and $v'$ is attached to $w'\neq w$, then $v,v',w,w'$ form a $2k_2$ and we can find it in $O(n^2)$ time. So suppose that for all $v'\in S$ with $d_T(v')=1$, $v'$ is attached to $w$ in $T$, and is attached to no other nodes in $T$. Similar to $v$, we can assume that for any $v'\in S$ with $d_T(v')=1$, $v'$ is not attached to $c$ (otherwise we find a $C_4$ in $O(n)$ time). So this means that $\tilde{C}=C\cup \{w\}\setminus \{c\}$ is a clique and $\tilde{T}=T\cup S\cup\{c\}\setminus \{w\}$ is an independent set and $\tilde{C}\cup \tilde{T}=V(G)$. So by observation \ref{obs:clique-or-ind}, $G$ has no $C_4$ or $2k_2$.

Now it suffices to solve the problem for the case where $|T|\le 2$ and there is $v\in S$ where $d_T(v)=2$. So $|T|=2$. If there is another node $v'\in S$, $v'\neq v$, such that $d_T(v)=2$, then $v,v',T$ form a $C_4$. So for all $v\neq v'\in S$, $d_T(v')\le 1$. Let $T=\{t_1,t_2\}$. If there is a $c\in C$, such that both $t_1$ and $t_2$ are not attached to $c$, then there is $c\neq c'\in C$, such that $v$ is not attached to $c'$ and $c'$ is attached to both $t_1,t_2$, and thus $c',v,t_1,t_2$ form a $C_4$. 

So suppose that there are $c_1,c_2\in C$, such that $t_1c_1\notin E(G)$ and $t_2c_2\notin E(G)$.
Note that from above we know that if $v$ has at most $|C|-3$ neighbors in $C$ we can find a $C_4$ or $2k_2$. So we can assume that $v$ is attached to all but exactly two nodes in $C$. If $v$ is attached to $c_1$, then $v,c_1,c_2,t_1$ form a $C_4$. If $v$ is attached to $c_2$, then $v,c_2,c_1,t_2$ form a $C_4$. So we can assume that $v$ is not attached to $c_1,c_2$, and is attached to all other nodes in $C$.

Now we are going to see if we find any $C_4$ or $2k_2$ containing $v$. First suppose that $v$ is in a $C_4$ and $u,w$ are its neighbors. Both $w,u$ cannot be in $C$, since $C$ is a clique. If $w\in C_4$ and $u\in T$, then $w\notin \{c_1,c_2\}$, so $w$ is attached to $u$. So this is not possible either. If $w,u\in T$, then the forth node must be in $C$, and since it cannot be $c_1$ or $c_2$, it is attached to $v$, which is again not possible. So there is no $C_4$ having $v$ as a node. 

Now suppose that there is a $2k_2$ containing $v$. Let the $k_2$ containing $v$ be $vu$. If $u\in T$, WLOG suppose $u=t_1$. Since $v$ is adjacent to $t_2$, the other $k_2$ must be fully in $C$ or have one node in $C$ and one node in $S$. If it is fully in $C$, then it must be $c_1c_2$. But $t_1$ is attached to $c_2$. So the other $k_2$ has one node in $C$ one node in $S$. In this case, the node in $C$ must be $c_1$ since $t_1$ is not attached to it. So we look at all the neighbors of $c_1$ in $S$, and see if they make a $2k_2$ with $c_1,t_2,v$, in $O(n)$ time. 

Now suppose that $u\in C$. This is not possible since all nodes in $T$ are adjacent to $v$, and all nodes in $C$ are adjacent to $u$, and so the other $k_2$ must be fully in $S$ which is an independent set. 

So we can remove $v$ from the graph, and now we don't have any node in $S$ that has two neighbors in $T$, and we can proceed as the case where all nodes in $S$ have at most one neighbor in $T$.
\end{proof}

\end{document}